\documentclass[a4paper,UKenglish,thm-restate]{lipics-v2021}
\pdfoutput=1 

\hideLIPIcs  

\usepackage[utf8]{inputenc}
\usepackage[T1]{fontenc}
\usepackage{microtype}
\usepackage{lmodern}

\usepackage{graphicx}
\usepackage{xcolor}
\usepackage[fleqn]{amsmath}
\usepackage{amsfonts}
\usepackage{amssymb}
\usepackage{amsthm}

\usepackage{tabularx}
\usepackage[strings]{underscore}

\usepackage{multicol}

\usepackage{pbox}

\usepackage{xcolor}
\usepackage{stmaryrd}

\newcommand{\N}{\mathbb{N}}

\def\O{\mathcal{O}}

\newcommand{\Abs}[1]{\mathopen|#1\mathclose|}

\definecolor{nicebg}{HTML}{f6f0e4}
\definecolor{nicered}{HTML}{7f0a13}
\definecolor{nicebgred}{HTML}{f2e7e8}
\definecolor{niceblue}{HTML}{104354}
\definecolor{nicebgblue}{HTML}{e8edee}
\definecolor{nicegreen}{HTML}{217516}
\definecolor{nicebggreen}{HTML}{e9f1e8}
\definecolor{nicepurple}{HTML}{884bab}
\definecolor{nicebgpurple}{HTML}{f3edf7}
\definecolor{niceorange}{HTML}{d27c11}
\definecolor{nicebgorange}{HTML}{fbf2e8}
\definecolor{nicepink}{HTML}{e95f9f}
\definecolor{nicebgpink}{HTML}{fdeff6}
\definecolor{niceredlight}{HTML}{c9888d}
\definecolor{nicebluelight}{HTML}{78a4b8}
\definecolor{nicegreenlight}{HTML}{76de68}
\definecolor{nicepurplelight}{HTML}{bc87db}
\definecolor{niceredbright}{HTML}{bd0310}
\definecolor{nicebgredbright}{HTML}{f9e6e8}
\definecolor{nicebluebright}{HTML}{197b9b}
\definecolor{nicebgbluebright}{HTML}{e8f2f5}

\newcommand{\TraNs}{tra}
\newcommand{\TraName}[1]{\tag*{⟨\textsf{#1}⟩}\label{\TraNs:#1}}
\newcommand{\TraRef}[1]{\text{\ref{\TraNs:#1}}}

\usepackage{etoolbox}
\makeatletter
\pretocmd\start@gather{%
    \if@minipage\kern-\topskip\kern-\baselineskip\kern+7pt\fi
}{}{}
\makeatother

\newcommand{\parag}[1]{\noindent\textbf{\textsf{#1.}}}

\usepackage{booktabs}
\usepackage{enumitem}

\usepackage{algorithm}
\usepackage[noend]{algpseudocode}
\algrenewcommand\textproc{\textsf}
\newcommand{\True}{\textnormal{\textcolor{niceblue}{\textsf{true}}}}
\newcommand{\False}{\textnormal{\textcolor{niceblue}{\textsf{false}}}}
\newcommand{\TrueNC}{\textnormal{\textsf{true}}}
\newcommand{\FalseNC}{\textnormal{\textsf{false}}}
\algrenewcommand\algorithmicindent{1.2em}
\algnewcommand\Input{\item[\textbf{Parameter:}]}
\algnewcommand\Output{\item[\textbf{Output:}]}
\algnewcommand\Effect{\item[\textbf{Effect:}]}
\allowdisplaybreaks

\newenvironment{namedalgorithm}[2][]{%
\begin{algorithm}[#1]\renewcommand{\thealgorithm}{\textnormal{\textproc{#2}}}%
\def\Name{#2}}{\end{algorithm}}

\makeatletter
\renewcommand{\fnum@algorithm}{\fname@algorithm {} \thealgorithm .}
\makeatother

\newcommand{\Reserve}{\mathsf{R}}

\def\O{\mathcal{O}}
\newcommand{\Restart}{\textnormal{\textbf{restart}}}

\newcommand{\Procs}{\mathrm{Proc}}
\newcommand{\Swap}{\textbf{swap}}

\newcommand{\Post}{\operatorname{post}}
\newcommand{\Prog}{\mathcal{P}}
\newcommand{\Rel}{\mathbin{\sim}}
\newcommand{\Confs}{\mathcal{C}}

\newcommand{\Flags}{F}
\newcommand{\Flagdom}{\mathcal{F}}
\newcommand{\IP}{\mathit{IP}}
\newcommand{\CF}{\mathit{CF}}
\newcommand{\OF}{\mathit{OF}}
\newcommand{\Maybe}{\textnormal{\textbf{detect}}}

\newcommand{\Machine}{\mathcal{A}}
\newcommand{\Virt}[1]{V_{#1}}
\newcommand{\Prot}{\mathit{PP}}
\newcommand{\Inst}{\mathcal{I}}

\newcommand{\PPInput}{I}
\newcommand{\Fstate}[1]{\textsf{#1}}
\newcommand{\Fstates}{S}

\newcommand{\Counter}{\operatorname{ctr}}

\title{Breaking through the $\Omega(n)$-space barrier: Population Protocols Decide Double-exponential Thresholds}
\titlerunning{Population Protocols Decide Double-exponential Thresholds}

\author{Philipp Czerner}{Department of Informatics, TU München, Germany \and \url{https://nicze.de/philipp} }{czerner@in.tum.de}{https://orcid.org/0000-0002-1786-9592}{This work was supported by an ERC Advanced Grant (787367: PaVeS) and by the Research Training Network of the Deutsche Forschungsgemeinschaft (DFG) (378803395: ConVeY).}

\authorrunning{P. Czerner}
\Copyright{Philipp Czerner}


\category{} 

\relatedversion{} 

\acknowledgements{}

\nolinenumbers 

\EventEditors{Dan Alistarh}
\EventNoEds{1}
\EventLongTitle{38th International Symposium on Distributed Computing (DISC 2024)}
\EventShortTitle{DISC 2024}
\EventAcronym{DISC}
\EventYear{2024}
\EventDate{October 28--November 1, 2024}
\EventLocation{Madrid, Spain}
\EventLogo{}
\SeriesVolume{319}
\ArticleNo{16}

\begin{document}

\maketitle

\begin{abstract}
Population protocols are a model of distributed computation in which finite-state agents interact randomly in pairs. A protocol decides for any initial configuration whether it satisfies a fixed property, specified as a predicate on the set of configurations. A family of protocols deciding predicates $\varphi_n$ is \emph{succinct} if it uses $\mathcal{O}(|\varphi_n|)$ states, where $\varphi_n$ is encoded as quantifier-free Presburger formula with coefficients in binary. (All predicates decidable by population protocols can be encoded in this manner.) While it is known that succinct protocols exist for all predicates, it is open whether protocols with $o(|\varphi_n|)$ states exist for \emph{any} family of predicates $\varphi_n$. We answer this affirmatively, by constructing protocols with $\mathcal{O}(\log|\varphi_n|)$ states for some family of threshold predicates $\varphi_n(x)\Leftrightarrow x\ge k_n$, with $k_1,k_2,...\in\N$. (In other words, protocols with $\mathcal{O}(n)$ states that decide $x\ge k$ for a $k\ge 2^{2^n}$.) This matches a known lower bound. Moreover, our construction for threshold predicates is the first that is not $1$-aware, and it is almost self-stabilising.

\end{abstract}

\section{Introduction}
Population protocols are a distributed model of computation where a large number of indistinguishable finite-state agents interact randomly in pairs. The goal of the computation is to decide whether an initial configuration satisfies a given property. The model was introduced in 2004 by Angluin et al.~\cite{AngluinADFP04,AngluinADFP06} to model mobile sensor networks with limited computational capabilities (see e.g.~\cite{PerronVV09,DraiefV12}). It is also closely related to the model of chemical reaction networks, in which agents, representing discrete molecules, interact stochastically~\cite{ChenCDS17}.

A protocol is a finite set of transition rules according to which agents interact, but it can be executed on an infinite family of initial configurations. Agents decide collectively whether the initial configuration fulfils some (global) property by \emph{stable consensus}; each agent holds an opinion about the output and may freely change it, but eventually all agents agree.

An example of a property decidable by population protocols is \emph{majority}: initially all agents are in one of two states, $x$ and $y$, and they try to decide whether $x$ has at least as many agents as $y$. This property may be expressed by the predicate $\varphi(x,y)\Leftrightarrow x\ge y$.

In a seminal paper, Angluin et al.~\cite{AngluinAER07} proved that the predicates that can be decided by population protocols correspond precisely to the properties expressible in Presburger arithmetic, the first-order theory of addition.

To execute a population protocol, the scheduler picks two agents uniformly at random and executes a pairwise transition on these agents.  These two agents interact and may change states. The number of agents does not change during the computation. It will be denoted $m$ throughout this paper.



Population protocols are often extended with a \emph{leader} ― an auxiliary agent not part of the input, which can assist the computation. It is known that this does not increase the expressive power of the model, i.e.\ it can still decide precisely the predicates expressible in Presburger arithmetic. However, it is known that leaders enable an exponential speed-up~\cite{AngluinAE08a,AlistarhAEGR17} in terms of the time that is needed to come to a consensus.

\parag{Space complexity}
Many constructions in the literature need a large number of states. We estimate, for example, that the protocols of~\cite{AngluinAE08a} need tens of thousands of states. This is a major obstacle to implementing these protocols in chemical reactions, as every state corresponds to a chemical compound.

\newcommand{\Space}{\operatorname{space}}
\newcommand{\Poly}{\operatorname{poly}}
This motivates the study of \emph{space complexity}, the minimal number of states necessary for a population protocol to decide a given predicate. Predicates are usually encoded as quantifier-free Presburger formulae with coefficients in binary. For example, the predicates $\varphi_n(x)\Leftrightarrow x\ge 2^n$ have length $\Abs{\varphi_n}\in\Theta(n)$. Formally we define $\Space(\varphi)$ as the smallest number of states of any protocol deciding $\varphi$, and $\Space_L(\varphi)$ as the analogous function for protocols with a leader. Clearly, $\Space(\varphi)_L\le\Space(\varphi)$.

The original construction in~\cite{AngluinADFP04} showed $\Space(\varphi)\in\O(2^{\Abs{\varphi}})$ – impractically large. For the family of \emph{threshold predicates} $\tau_n(x)\Leftrightarrow x\ge n$ Blondin, Esparza and Jaax~\cite{BlondinEJ18} prove $\Space(\tau_n)\in\O(\Abs{\tau_n})$, i.e.\ they have polynomial space complexity. For several years it was open whether similarly succinct protocols exist for every predicate. This was answered positively in~\cite{BlondinEGHJ20}, showing $\Space(\varphi)\in\O(\Poly(\Abs{\varphi}))$ for all $\varphi$. 

Is it possible to do much better? For most predicates it is not; based on a simple counting argument one can show that for every family $\varphi_n$ with $\Abs{\varphi_n}\in\O(n)$ there is an infinite subfamily $(\varphi_n')_n\subseteq(\varphi_n)_n$ with $\Space_L(\varphi_n')\in\Omega(\Abs{\varphi_n}^{1/4-\varepsilon})$, for any $\varepsilon>0$~\cite{BlondinEJ18}.

This covers threshold predicates and many other natural families of protocols (e.g.\ $\varphi_n(x)\Leftrightarrow x\equiv 0\pmod n$ or $\varphi_n(x,y)\Leftrightarrow x\ge ny$). But it is not an impenetrable barrier, even for the case of threshold protocols: it does not rule out constructions that work for \emph{infinitely many} (but not all) thresholds and use only, say, logarithmically many states. Indeed, if leaders are allowed this is known to be possible: \cite{BlondinEJ18} shows $\Space_L(\tau_n')\in\O(\log \Abs{\tau_n'})$ for some subfamily $\tau_n'$ of threshold predicates.

Recently, \emph{general} lower bounds have been obtained, showing $\Space(\tau_n)\in\Omega(\log^{1-\varepsilon} \Abs{\tau_n})$ for all $\varepsilon>0$~\cite{CzernerE21,CzernerEL22}. The same bound (up to $\varepsilon=1/2$) holds even if the model is extended with leaders~\cite{Leroux21}.

For leaderless population protocols, these results leave an exponential gap. In this paper we settle that question and show that, contrary to prevailing opinion, $\Space(\tau_n')\in\O(\log \Abs{\tau_n'})$ for some subfamily $\tau_n'$ of threshold predicates. In other words, we construct the first family of \emph{leaderless} population protocols that decide double-exponential thresholds and break through the polynomial barrier.

\begin{table}
\caption{Prior results on the state complexity of threshold predicates $\varphi(x)\Leftrightarrow x\ge k$, for $k\in\N$. Upper bounds need only hold for infinitely many $k$. We elide exponentially dominated factors from lower bounds.}\label{tab:results}
\begin{center}
\begin{tabular}{lllll}
year & result & type & ordinary & with leaders \\\midrule
2018 & Blondin, Esparza, Jaax~\cite{BlondinEJ18} & construction & $\O(\Abs{\varphi})$ & $\O(\log\Abs{\varphi})$ \\
2021 & Czerner, Esparza~\cite{CzernerE21} & impossibility & $\Omega(\log\log\Abs{\varphi})$ & $\Omega(\operatorname{ack}^{-1}\Abs{\varphi})$ \\
2021 & Czerner, Esparza, Leroux~\cite{CzernerEL22}& impossibility & $\Omega(\log\Abs{\varphi})$ &\\
2022 & Leroux~\cite{Leroux21} & impossibility & &$\Omega(\log\Abs{\varphi})$ \\
2024& this paper & construction & $\O(\log\Abs{\varphi})$ &\\
\end{tabular}
\end{center}
\end{table}

\parag{Robustness}
Since population protocols model computations where large numbers of agents interact, it is desirable that protocols deal robustly with noise. In a chemical reaction, for example, there can be trace amounts of unwanted molecules. So the initial configuration of the protocol would have the form $C_I+C_N$, where $C_I$ is the “intended” initial configuration, containing only agents in the designated initial states, and $C_N$ is a “noise” configuration, which can contain agents in arbitrary states.

For threshold predicates, specifically, we want to decide whether $\Abs{C_I}+\Abs{C_N}$ exceeds some threshold $k\in\N$, under some reasonable restrictions to $C_I,C_N$. However, all known threshold protocols fail even for the case $\Abs{C_N}=1$. Is it possible to do better?

If $C_N$ can be chosen arbitrarily, then the protocol has to work correctly for \emph{all} input configurations. This property is known as \emph{self-stabilisation}, and it has also been investigated in the context of population protocols~\cite{AngluinAFJ05,CaiIW12,BurmanCCDNSX21}. However, it can only be achieved in extensions of the model (e.g.\ on specific communication graphs, or with a non-constant number of states). This is easy to see in the case of threshold predicates: if any configuration is stably accepting, then any smaller configuration is stably accepting as well. In particular, there is a stably accepting configuration with $k-1$ agents.

While full self-stabilisation is impossible, in this paper we show that one can come remarkably close. We prove that our construction is \emph{almost self-stabilising}, meaning that it computes the correct output for all $C_I,C_N$ with $\Abs{C_I}\ge n$, where $n$ is the number of states of the protocol. We do not constraint $C_N$ at all. Since $n\in\O(\log\log k)$ in our protocol, this means that one can take an arbitrary configuration $C_N$ one wishes to count, add a tiny amount of agents to the initial state, and the protocol will compute the correct output.

\parag{Related work}
We consider the space complexity of families of protocols, each of which decides a different predicate. In another line of research, one considers a family of protocols for the \emph{same} predicate, where each protocol is specialised for a fixed population size $m$.

In the original model of population protocols (which is also the model of this paper), the set of states is fixed, and the same protocol can be used for an arbitrary number of agents. Relaxing this requirement has opened up a fruitful avenue of research; here, the number of states depends on $m$ (e.g.\ the protocol has $\O(\log m)$ states, or even $\O(\log\log m)$ states). In this model, faster protocols can be achieved~\cite{AlistarhGV15,MocquardAABS15,MocquardAS16}.

It has also led to space-efficient, fast protocols, which stabilise within $\O(\operatorname{polylog}m)$ parallel time, using a state-space that grows only slowly with the number of agents, e.g.\ $\O(\operatorname{polylog}m)$ states~\cite{AlistarhAEGR17,BilkeCER17,AlistarhAG18,BerenbrinkEFKKR18,NunKKP20,BerenbrinkEFKKR21,DotyEGSUS21}. These protocols have focused on the majority predicate.
Moreover, lower bounds and results on time-space tradeoffs have been developed in this model~\cite{AlistarhAEGR17,AlistarhAG18}.

\section{Main result}\label{sec:mainresult}
We construct population protocols (without leaders) for an infinite family of threshold predicates $\varphi_n(x)\Leftrightarrow x\ge k_n$, with $k_1,...\in\N$, proving an $\O(\log\Abs{\varphi_n})$ upper bound on their state complexity. This closes the final gap in the state complexity of threshold predicates.

As in prior work, our result is not a construction for arbitrary thresholds $k$, only for an infinite family of them. It is, therefore, easier to formally state by fixing the number of states $n$ and specifying the largest threshold $k$ that can be decided by a protocol with $n$ states.

\begin{theorem}\label{thm:main}
For every $n\in\N$ there is a population protocol with $\O(n)$ states deciding the predicate $\varphi(x)\Leftrightarrow x\ge k$ for some $k\ge 2^{2^n}$.
\end{theorem}
\begin{proof}
This will follow from theorems~\ref{thm:exists-program} and~\ref{thm:conversion}.
\end{proof}

The result is surprising, as prevailing opinion was that the existing constructions are optimal. This was based on the following:
\begin{itemize}
\item It is intuitive that population protocols with leaders have an advantage. In particular, one can draw a parallel to time complexity, where an exponential gap is proven: for some predicates protocols with leaders have $\O(\operatorname{polylog} m)$ parallel time, while all leaderless protocols have $\Omega(m)$ parallel time.
\item The $\O(\log\log k)$-state construction from~\cite{BlondinEJ18} crucially depends on having leaders.
\item The technique to show the $\Omega(\log\log k)$ lower bound could, for the most part, also be used for a $\Omega(\log k)$ bound. Only the use of Rackoff's theorem, a general result for Petri nets, does not extend.
\item There is a conditional impossibility result, showing that $\Omega(\log k)$ states are necessary for leaderless $1$-aware protocols.~\cite{BlondinEJ18} (Essentially, protocols where some agent knows at some point that the threshold has been exceeded.) All prior constructions are $1$-aware.
\end{itemize}

Regarding the last point, our protocol evades the mentioned conditional impossibility result by being the first construction that is not $1$-aware. Intuitively, our protocol only accepts provisionally and continues to check that no invariant has been violated. Based on this, we also obtain the following robustness guarantee:

\begin{restatable}{theorem}{restateMainTwo}\label{thm:main2}
The protocols of Theorem~\ref{thm:main} are almost self-stabilising.
\end{restatable}

\parag{Overview}
We build on the technique of Lipton~\cite{lipton1976reachability}, which describes a double-exponential counting routine in vector addition systems. Implementing this technique requires the use of procedure calls; our first contribution are \emph{population programs}, a model in which population protocols can be constructed by writing structured programs, in Section~\ref{sec:population-programs}. Every such program can be converted into an equivalent population protocol.

However, population programs provide weaker guarantees than the model of parallel programs used in~\cite{lipton1976reachability}. Both models access registers with values in $\N$. In a parallel program these are initialised to 0, while in a population program \emph{all} registers start with arbitrary values. This limitation is essential for our conversion into population protocols. 

A straightforward implementation is, therefore, impossible. Instead, we have to adapt the technique to work with arbitrary initial configurations. Our second contribution, and the main technical difficulty of this result, is extending the original technique with error-checking routines to work in our model. We use a detect-restart loop, which determines whether the initial configuration is “bad” and, if so, restarts with a new initial configuration. The stochastic behaviour of population protocols ensures that a “good” initial configuration is reached eventually. Standard techniques could be used to avoid restarts with high probability and achieve an optimal running time, but this is beyond the scope of this paper.

A high level overview of both the original technique as well as our error-checking strategy is given in Section~\ref{sec:overview}. We then give a detailed description of our construction in Section~\ref{sec:details}.

To get population protocols, we need to convert from population programs. We split this into two parts. First, we use standard techniques to lower population programs to \emph{population machines}, an assembly-like programming language. In a second step we simulate arbitrary population machines by population protocols. This conversion is described in Section~\ref{sec:conversion}.

Finally, we introduce the notion of being almost self-stabilising in Section~\ref{sec:robustness}, and prove that our construction has this property.

To start out, Section~\ref{sec:preliminaries} introduces the necessary mathematical notation and formally defines population protocols as well as the notion of stable computation.

\section{Preliminaries}\label{sec:preliminaries}

\parag{Multisets} We assume $0\in\N$. For a finite set $Q$ we write $\N^Q$ to denote the set of multisets containing elements in $Q$. For such a multiset $C\in\N^Q$, we write $C(S):=\sum_{q\in S}C(q)$ to denote the total number of elements in some $S\subseteq Q$, and set $\Abs{C}:=C(Q)$. Given two multisets $C,C'\in\N^Q$ we write $C\le C'$ if $C(q)\le C'(q)$ for all $q\in Q$, and we write $C+C'$ and $C-C'$ for the componentwise sum and difference (the latter only if $C\ge C'$). Abusing notation slightly, we use an element $q\in Q$ to represent the multiset $C$ containing exactly $q$, i.e.\ $C(q)=1$ and $C(r)=0$ for $r\ne q$. 

\parag{Stable computation} We are going to give a general definition of stable computation not limited to population protocols, so that we can later reuse it for population programs and population machines. Let $\Confs$ denote a set of configurations and $\rightarrow$ a left-total binary relation on $\Confs$ (i.e.\ for every $C\in\Confs$ there is a $C'\in\Confs$ with $C\rightarrow C'$). Further, we assume some notion of output, i.e.\ some configurations have an output $b\in\{\True,\False\}$ (but not necessarily all).

A sequence $\tau=(C_i)_{i\in\N}$ with $C_i\in\Confs$ is a \emph{run} if $C_i\rightarrow C_{i+1}$ for all $i\in\N$. We say that $\tau$ \emph{stabilises to $b$}, for $b\in\{\True,\False\}$, if there is an $i$ s.t.\ $C_j$ has output $b$ for every $j\ge i$. A run $\tau$ is \emph{fair} if $\cap_{i\ge0}\{C_i,C_{i+1},...\}$ is closed under $\rightarrow$, i.e.\ every configuration that \emph{can} be reached infinitely often \emph{is}.

\parag{Population protocols} A \emph{population protocol} is a tuple $\Prot=(Q,\delta,I,O)$, where
\begin{multicols}{2}
\begin{itemize}
\item $Q$ is a finite set of \emph{states},
\item $\delta\subseteq Q^4$ is a set of \emph{transitions},
\item $I\subseteq Q$ is a set of \emph{input states}, and
\item $O\subseteq Q$ is a set of \emph{accepting states}.
\end{itemize}
\end{multicols}
We write transitions as $(q,r\mapsto q',r')$, for $q,r,q',r'\in Q$. A \emph{configuration} of $\Prot$ is a multiset $C\in\N^Q$ with $\Abs{C}>0$. A configuration $C$ is \emph{initial} if $C(q)=0$ for $q\notin I$ (one might also say $C\in\N^I$ instead). It has output \True\ if $C(q)=0$ for $q\notin O$, and output \False\ if $C(q)=0$ for $q\in O$. For two configurations $C,C'$ we write $C\rightarrow C'$ if $C=C'$ or if there is a transition $(q,r\mapsto q',r')\in\delta$ s.t.\ $C\ge q+r$ and $C'=C-q-r+q'+r'$.

Let $\varphi:\N^I\rightarrow\{\True,\False\}$ denote a predicate. We say that $\Prot$ \emph{decides} $\varphi$, if every fair run starting at an initial configuration $C\in\N^I$ stabilises to $\varphi(C)$, where fair run and stabilisation are defined as above.


\section{Population Programs}\label{sec:population-programs}
We introduce population programs, which allows us to specify population protocols using structured programs. An example is shown in Figure~\ref{fig:program-example}.

Formally, a \emph{population program} is a tuple $\Prog=(Q,\Procs)$, where $Q$ is a finite set of \emph{registers} and $\Procs$ is a list of \emph{procedures}. Each procedure has a name and consists of (possibly nested) while-loops, if-statements and instructions. These are described in detail below.

\begin{figure}[t]
\algrenewcommand\algorithmicindent{0.8em}
\begin{minipage}[t]{0.32\textwidth}
\begin{algorithmic}[1]
\Procedure{Main}{}
    \State $\OF:=\False$
    \While{$\neg\Call{Test}{4}$}
        \State\Call{Clean}{}
    \EndWhile
    \State $\OF:=\True$
    \While{$\neg\Call{Test}{7}$}
        \State\Call{Clean}{}
    \EndWhile
    \State $\OF:=\False$
    \While{\True}
        \State\Call{Clean}{}
    \EndWhile
\EndProcedure
\end{algorithmic}
\end{minipage}
\begin{minipage}[t]{0.32\textwidth}
\begin{algorithmic}[1]
\Procedure{Test}{$i$}
    \For{$j=1,...,i$}
        \If{\Maybe\ $x>0$}
            \State $x\mapsto y$
        \Else
            \State\Return\False
        \EndIf
    \EndFor
    \State\Return\True
\EndProcedure
\end{algorithmic}
\end{minipage}
\begin{minipage}[t]{0.32\textwidth}
\begin{algorithmic}[1]
\Procedure{Clean}{}
    \If{\Maybe\ $z>0$}
        \State\Restart
    \EndIf
    \State\Swap\ $x,y$
    \While{$\Maybe\ y>0$}
        \State $y\mapsto x$
    \EndWhile
\EndProcedure
\end{algorithmic}
\end{minipage}
\caption{A population program for $\varphi(x)\Leftrightarrow 4\le x<7$ using registers $x,y,z$.
\Call{Main}{} is run initially and decides the predicate, $\Call{Test}{i}$ tries to move $i$ units from $x$ to $y$ and reports whether it succeeded, and \Call{Clean}{} checks whether $z$ is empty and moves some number of units from $y$ to $x$. If \Call{Clean}{} detects an agent in $z$, it restarts the computation. As every run calls \Call{Clean}{} infinitely often, this serves to reject initial configurations where $z$ is nonzero; eventually the protocol will be restarted with $z=0$.
This is an illustrative example and some simplifications are possible. E.g.\ the instruction $(\Swap\ x,y)$ in \Call{Clean}{} is superfluous; additionally, instead of checking $z>0$ one could omit that register entirely.
}
\label{fig:program-example}
\end{figure}

\parag{Primitives} Each register $x\in Q$ can take values in $\N$. Only three operations on these registers are supported.
\begin{itemize}
\item The move instruction $(x\mapsto y)$, for $x,y\in Q$, decreases the value of $x$ by one, and increases the value of $y$ by one. We also say that it moves one unit from $x$ to $y$. If $x$ is empty, i.e.\ its value is zero, the programs hangs and makes no further progress
\item The nondeterministic nonzero-check $(\Maybe\ x>0)$, for $x\in Q$, nondeterministically returns either $\False$ or whether $x>0$. In other words, if it does return \True, it certifies that $x$ is nonzero. If it returns \False, however, no information has been gained. We consider only fair runs, so if $x$ is nonzero the check cannot return \False\ infinitely often.
\item A swap $(\Swap\ x,y)$ exchanges the values of the two registers $x,y$. This primitive is not necessary, but it simplifies the implementation.
\end{itemize}

\parag{Loops and branches} Population programs use while-loops and if-statements, which function as one would expect.

We also use for-loops. These, however, are just a macro and expand into multiple copies of their body. For example, in the program in Figure~\ref{fig:program-example} the for-loop in \Call{Test}{} expands into $i$ copies of the contained if-statement.

\parag{Procedures} Our model has procedure calls, but no recursion. Procedures have no arguments, but we may have parameterised copies of a procedure. The program in Figure~\ref{fig:program-example}, for example, has four procedures: \Call{Main}{}, \Call{Clean}{}, \Call{Test}{4}, and \Call{Test}{7}.

Procedure calls must be acyclic. It is thus not possible for a procedure to call itself, and the size of the call stack remains bounded. We remark that one could inline every procedure call. The main reason to make use of procedures at all is succinctness: if our program contains too many instructions, the resulting population protocol has too many states.

Procedures may return a single boolean value, and procedure calls can be used as expressions in conditions of while- or if-statements.

\parag{Output flag} There is an output flag $\OF$, which can be modified only via the instructions $\OF:=\True$ and $\OF:=\False$. (These are special instructions; it is not possible to assign values to registers.) The output flag determines the output of the computation. 

\parag{Initialisation and restarts} The only guarantee on the initial configuration is that execution starts at \Call{Main}{}. In particular, all registers may have arbitrary values.

There is one final kind of instruction: \Restart. As the name suggests, it restarts the computation. It does so by nondeterministically picking any initial configuration s.t.\ the sum of all registers does not change.

\parag{Size} The \emph{size} of $\Prog$ is defined as $\Abs{Q}+L+S$, where $L$ is the number of instructions and $S$ is the \emph{swap-size}. The latter is defined as the number of pairs $(x,y)\in Q^2$ for which it is syntactically possible for $x$ to swap with $y$ via any sequence of swaps. \footnote{Unfortunately, without restrictions we would convert swaps to population protocols with a quadratic blow-up in states, so we introduce this technical notion to quantify the overhead.} For example, in Figure~\ref{fig:program-example} the swap-size is two: $(x,y),(y,x)$ can be swapped, but e.g.\ $(x,z)$ cannot. If we add a (\textbf{swap} $y,z$) instruction at any point, then $(x,z)$ can be swapped (transitively), and the swap-size would be $6$.

\parag{Configurations and Computation} A \emph{configuration} of $\Prog$ is a tuple $D=(C,\OF,\sigma)$, where $C\in\N^Q$ is the \emph{register configuration}, $\OF\in\{\True,\False\}$ is the value of the output flag, and $\sigma\in(\Procs\times\N)^*$ is the call stack, storing names and currently executed instructions of called procedures. (E.g.\ $\sigma=((\Call{Main}{},3),(\Call{Test(4)}{},1))$ when \Call{Test}{} is first called in Figure~\ref{fig:program-example}.) A configuration is \emph{initial} if $\sigma=((\Call{Main}{},1))$ and it has \emph{output} $\OF$. For two configurations $D,D'$ we write $D\rightarrow D'$ if $D$ can move to $D'$ after executing one instruction. 

Using the general notion of stable computation defined in Section~\ref{sec:preliminaries}, we say that $\Prog$ \emph{decides} a predicate $\varphi(x)$, for $k\in\N$, if every run started at an initial configuration $(C,\OF,\sigma)$ stabilises to $\varphi(|C|)$.
Note that this definition limits population programs to decide only unary predicates.

\parag{Notation}
When analysing population programs it often suffices to consider only the register configuration Let $C,C'\in\N^Q$, $b\in\{\False,\True\}$ and let $f\in\Procs$ denote a procedure. We consider the possible outcomes when executing $f$ in a configuration with registers $C$. Note that the program is nondeterministic, so multiple outcomes are possible. If $f$ may return $b$ with register configuration $C'$, we write $C,f\rightarrow C',b$. For procedures not returning a value, we use $C,f\rightarrow C'$ instead. If $f$ may initiate a restart, we write $C,f\rightarrow\Restart$. If $f$ may hang or not terminate, we write $C,f\rightarrow\bot$. Finally, we define $\Post(C,f):=\{S:C,f\rightarrow S\}$.

\section{High-level Overview}\label{sec:overview}
We give an intuitive explanation of our construction. This section has two parts. As mentioned, we use the technique of Lipton~\cite{lipton1976reachability} to count to $2^{2^n}$ using $4n$ registers. We will give a brief explanation of the original technique in Section~\ref{ssec:de-counting}. Readers might also find the restatement of Liptons proof in~\cite{Esparza96} instructive ― the Petri net programs introduced therein are closer to our approach, and more similar to models used in the recent Petri net literature.

A straightforward application of the above technique only works if some guarantees are provided for the initial configuration (e.g.\ that the $4n$ registers used are empty, while an additional register holds all input agents). No such guarantees are given in our model. Instead, we have to deal with adversarial initialisation, i.e.\ the notion that registers hold arbitrary values in the initial configuration. Section~\ref{ssec:error-detection} describes the problems that arise, as well as our strategies for dealing with them.

\subsection{Double-exponential counting}\label{ssec:de-counting}
The biggest limitation of population programs is their inability to detect absence of agents. This is reflected in the $(\Maybe\ x>0)$ primitive; it may return $\True$ and thereby certify that $x$ is nonzero, but it may always return $\False$, regardless of whether $x=0$ actually holds. In particular, it is impossible to implement a zero-check.

However, Lipton observes that if we have two registers $x,\overline{x}$ and ensure that the invariant $x+\overline{x}=k$ holds, for some fixed $k\in\N$, then $x=0$ is equivalent to $\overline{x}\ge k$. Crucially, it is possible to certify the latter property; if we have a procedure for checking $\overline{x}\ge k$, we can run both checks ($x>0$ and $\overline{x}\ge k$) in a loop until one of them succeeds. Therefore, we may treat $x$ as $k$-bounded register with deterministic zero-checks. 

This seems to present a chicken-and-egg problem: to implement this register we require a procedure for $\overline{x}\ge k$, but checking such a threshold is already the overall goal of the program. Lipton solves this by implementing a bootstrapping sequence. For small $k$, e.g.\ $k=2$, one can easily implement the required $\overline{x}\ge k$ check. We use that as subroutine for \emph{two} $k$-bounded registers, $x$ and $y$. Using the deterministic zero-checks, $x$ and $y$ can together simulate a single $k^2$-bounded register with deterministic zero-check; this then leads to a procedure for checking $\overline{z}\ge k^2$ (for some other register $\overline{z}$).

Lipton iterates this construction $n$ times. We have $n$ levels of registers, with four registers $x_i,y_i,\overline{x}_i,\overline{y}_i$ on each level $i\in\{1,...,n\}$. For each level we have a constant $N_i\in\N$ and ensure that $x_i+\overline{x}_i=y_i+\overline{y}_i=N_i$ holds. These constants grow by repeated squaring, so e.g.\ $N_1=2$ and $N_{i+1}=N_i^2$. Clearly, $N_n=2^{2^n}$. (Our actual construction uses slightly different $N_i$.)

We have not yet broached the topic of initialising these registers s.t.\ the necessary invariants hold. For our purposes, having a separate initialisation step is superfluous. Instead, we check whether the invariants hold in the initial configuration and restart (nondeterministically choosing a new initial configuration) if they do not.

\subsection{Error detection}\label{ssec:error-detection}
Our model provides only weak guarantees. In particular, we must deal with adversarial initialisation, meaning that the initial configuration can assign arbitrary values to any register. This is not limited to a designated set of initial registers; all registers used in the computation are affected.

Let us first discuss how the above construction behaves if its invariants are violated. As above, let $x,\overline{x}$ denote registers for which we want to keep the invariant $x+\overline{x}=k$, for some $k\in\N$. If instead $x+\overline{x}>k$, the “zero-check” described above is still guaranteed to terminate, as either $x>0$ or $\overline{x}\ge k$ must hold. However, it might falsely return $x=0$ when it is not. The procedure we use above, to combine two $k$-bounded counter to simulate a $k^2$-bounded counter, exhibits erratic behaviour under these circumstances. When we try to use it to count to $k^2$ we might instead only count to some lower value $k'<k^2$, even $k'\in\O(k)$.

If the invariant is violated in the other direction, i.e.\ $x+\overline{x}<k$ holds, we can never detect $x=0$ and will instead run into an infinite loop.

The latter case is more problematic, as detecting it would require detecting absence. For the former, we can ensure that we check $x+\overline{x}\ge k+1$ infinitely often; if $x+\overline{x}>k$, this check will eventually return $\True$ and we can initiate a restart. For the $x+\overline{x}>k$ case the crucial insight is that we cannot \emph{detect} it, but we can \emph{exclude} it: we issue a single check $x+\overline{x}\ge k$ in the beginning. If it fails, we restart immediately.

\parag{A simplified model}
In the full construction, we have many levels of registers that rely on each other. Instead, we first consider a simplified model here to explain the main ideas.

{
\def\Check{\operatorname{\textsc{Check}}}
In our simplified model there is only a single register $x_i$ per level $i\in\{1,...,n\}$ as well as one “level $n+1$” register $\Reserve$. For $i\in\{1,...,n\}$ we are given subroutines $\Check(x_i\ge N_i)$ and $\Check(x_i>N_i)$ which we use to check thresholds; however, they are only guaranteed to work if $x_1=N_1$, $x_2=N_2$, ..., $x_{i-1}=N_{i-1}$ hold.

Our goal is to decide the threshold predicate $m\ge \sum_iN_i$, where $m:=\sum_i x_i+\Reserve$ is the sum of all registers. For each possible value of $m$ we pick one initial configuration $C_m$ and design our procedure s.t.
\begin{itemize}
\item every initial configuration different from $C_m$ will cause a restart, and
\item if started on $C_m$ it is \emph{possible} that the procedure enters a state where it cannot restart.
\end{itemize}
The structure of $C_m$ is simple: we pick the largest $i$ s.t.\ we can set $x_j:=N_j$ for $j\le i$ and put the remaining units into $x_{i+1}$ (or $\Reserve$, if $i=n$). The procedure works as follows:
\begin{enumerate}
\item We nondeterministically guess $i\in\{0,...,n\}$.
\item We run $\Check(x_j\ge N_j)$ for all $j\in\{1,...,i\}$. If one of these checks fails, we restart.
\item According to $i=n$ we set the output flag to $\True$ or $\False$.
\item To verify that we are in $C_m$, we check the following infinitely often. For $j\in\{1,...,i\}$ we run $\Check(x_j>N_j)$ and restart if it succeeds. If $i<n$ we also restart if $\Check(x_{i+1}\ge N_{i+1})$ or one of $x_{i+2}, ..., x_n, \Reserve$ is nonempty.
\end{enumerate}

Clearly, when started in $C_m$ and $i$ is guessed correctly, it is possible for step 2 to succeed, and it is impossible for step 4 to restart. If $i$ is too large, step 2 cannot work, and if $i$ is too small step 4 will detect $x_{i+1}\ge N_{i+1}$. So the procedure will restart until the right $i$ is guessed and step 4 is reached.

Consider an initial configuration $C\ne C_m$, $|C|=m$. There are two cases: either there is a $k$ with $C(x_k)<C_m(x_k)$, or some $k$ has $C(x_k)>C_m(x_k)$. Pick a minimal such $k$.

In the former case, step 2 can only pass if $i<k$, but then one of $x_{i+2}, ..., x_n, \Reserve$ is nonempty and step 4 will eventually restart.

The latter case is more problematic. Step 2 can pass regardless of $i$ (for $i>k$ the precondition of $\Check$ is not met). In step 4, either $i<k$ and then $x_{i+1}\ge N_{i+1}$ or one of $x_{i+2}, ..., x_n, \Reserve$ is nonempty, or $i\ge k$ and one of the checks $\Check(x_j>N_j)$ will eventually restart, for $j=k$.

This would be what we are looking for, but note that we implicitly made assumptions about the behaviour of $\Check$ when called without its precondition being met. We need two things: all calls to $\Check$ terminate and they do not change the values of any register. The second is the simpler one to deal with: later, we will have multiple registers per level and our procedures only need to move agents between registers of the same level. This keeps the sum of registers of one level constant, this weaker property suffices for correctness.

Ensuring that all calls terminate is more difficult. It runs into the problem discussed above, where a zero-check might not terminate if the invariant of its register is violated. In this simplified model it corresponds to the case $x_i<N_i$.

However, we note that $\Check(x_i\ge N_i)$ and $\Check(x_i>N_i)$ are only called if $(x_1,...,x_{i-1})\ge_{\mathrm{lex}}(N_1,...,N_{i-1})$, where $\ge_{\mathrm{lex}}$ denotes lexicographical ordering. So if the precondition is violated, there must be a $j<i$ with $(x_1,...,x_{j-1})=(N_1,...,N_{j-1})$ and $x_j>N_j$. This can be detected within the execution of $\Check$ by calling itself recursively. In this manner, we can implement $\Check$ in a way that avoids infinite loops as long as the weaker precondition $(x_1,...,x_{i-1})\ge_{\mathrm{lex}}(N_1,...,N_{i-1})$ holds.

Our actual construction follows the above closely; of course, instead of a single register per level we have four, making the necessary invariants more complicated. Additional issues arise when implementing $\Check$, as registers cannot be detected erroneous while in use. Certain subroutines must hence take care to ensure termination, even when the registers they use are not working properly.

}

\section{A Succinct Population Program}\label{sec:details}
In this section, we construct a population program $\Prog=(Q,\Procs)$ to prove the following:

\begin{restatable}{theorem}{restateExistsProgram}\label{thm:exists-program}
Let $n\in\N$. There exists a population program deciding $\varphi(x)\Leftrightarrow x\ge k$ with size $\O(n)$, for some $k\ge 2^{2^{n-1}}$.
\end{restatable}

Full proofs and formal definitions of this section can be found in Appendix~\ref{app:details}.
\smallskip

We use registers $Q:=Q_1\cup...\cup Q_n\cup \{\Reserve\}$, where $Q_i:=\{x_i,y_i,\overline{x}_i,\overline{y}_i\}$ are \emph{level $i$} registers and $\Reserve$ is a \emph{level $n+1$} register. For convenience, we identify $\overline{\overline{x}}$ with $x$ for any register $x$.
\smallskip

\parag{Types of Configurations}
As explained in the previous section, $x$ and $\overline{x}$ are supposed to sum to a constant $N_i$, for a level $i$ register $x\in\{x_i,y_i\}$, which we define via $N_1:=1$ and $N_{i+1}:=(N_i+1)^2$. If this invariant holds, we can use $x,\overline{x}$ to simulate a $N_i$-bounded register, which has value $x$. 

We cannot guarantee that this invariant always holds, so our program must deal with configurations that deviate from this. For this purpose, we classify configurations based on which registers fulfil the invariant, and based on the type of deviation.

A configuration $C\in\N^Q$ is \emph{$i$-proper}, if the invariant holds on levels $1,...,i$, and their simulated registers have value $0$. This is a precondition for most routines. Sometimes we relax the latter requirement on the level $i$ registers; $C$ is \emph{weakly $i$-proper} if it is $(i-1)$-proper and the invariant holds on level $i$.

If $C$ is $(i-1)$-proper and not $i$-proper, then there are essentially two possibilities. Either $C\le C'$ for some $i$-proper $C'$ and we call $C$ \emph{$i$-low}, or $C(x)\ge C'$ for a weakly $i$-proper $C'$ and we call $C$ \emph{$i$-high}. Note that it is possible that $C$ is neither $i$-low nor $i$-high — these configurations are easy to exclude and play only a minor role. We can mostly ensure that $i$-low configurations do not occur, but procedures must provide guarantees when run on $i$-high configurations.

Finally, we say that $C$ is \emph{$i$-empty} if all registers on levels $i,...,n+1$ are empty.

\begin{figure}[H]
\tabcolsep=3.5pt
\newcommand{\B}{\color{nicebluebright}}
\newcommand{\Bo}[1]{\B${#1}$}
\begin{tabular}{lccccccccccccccc}
&
$x_1$&$\overline{x}_1$&
$y_1$&$\overline{y}_1$&
...&
$x_{i-1}$&$\overline{x}_{i-1}$&
$y_{i-1}$&$\overline{y}_{i-1}$&&
$x_{i}$&$\overline{x}_{i}$&
$y_{i}$&$\overline{y}_{i}$&
... \\ \toprule
$i$-proper&
$0$&$N_1$&
$0$&$N_1$&
...&
$0$&$N_{i-1}$&
$0$&$N_{i-1}$&&
$0$&$N_{i}$&
$0$&$N_{i}$&
... \\
weakly $i$-proper&
$0$&$N_1$&
$0$&$N_1$&
...&
$0$&$N_{i-1}$&
$0$&$N_{i-1}$&&
\Bo{3}&\B $N_{i}-3$&
\B $N_{i}-7$&\Bo{7}&
... \\
$i$-low&
$0$&$N_1$&
$0$&$N_1$&
...&
$0$&$N_{i-1}$&
$0$&$N_{i-1}$&&
$0$&\B $N_{i}-3$&
$0$&$N_{i}$&
... \\
$i$-high&
$0$&$N_1$&
$0$&$N_1$&
...&
$0$&$N_{i-1}$&
$0$&$N_{i-1}$&&
\Bo{3}&$N_{i}$&
\B $7$&\B $N_i-5$&
... \\
$i$-empty&
\B$2$&\B$4$&
\B$8$&\B$3$&
...&
\B$5$&\B$3$&
$0$&\B$7$&&
$0$&\B$0$&
$0$&\B$0$&
... \\
\end{tabular}
\caption{Example configurations exhibiting the different types.}
\end{figure}

\parag{Summary} We use the following procedures.
\begin{itemize}
\item \ref{alg:main}. Computation starts by executing this procedure, and \ref{alg:main} ultimately decides the predicate $\varphi(x)\Leftrightarrow x\ge2\sum_{i=1}^nN_i$. 
\item \ref{alg:checkempty}. Check whether a configuration is $i$-empty and initiate a restart if not.
\item \ref{alg:checkproper}. Check whether a configuration is $i$-proper or $i$-low, initiate a restart if not.
\item \ref{alg:large}. Nondeterministically check whether a register $x\in Q_i$ is at least $N_i$. 
\item \ref{alg:zero}. Perform a deterministic zero-check on a register $x\in Q_i$.
\item \ref{alg:incrpair}. As described in Section~\ref{ssec:de-counting}, we use two level $i$ registers (which are $N_i$ bounded) to simulate an $N_{i+1}$-bounded register. This procedure implements the increment operation for the simulated register.
\end{itemize}

\begin{figure}[t]
\begin{minipage}[t]{0.51\textwidth}
\begin{namedalgorithm}[H]{AssertEmpty}
\caption{
}\label{alg:checkempty}
\begin{algorithmic}[1]
\Input $i\in\{1,...,n+1\}$
\Effect If $i$-empty, do nothing, else it may restart
\Procedure{\Name}{$i$} [$i\le n$]
\State $\Call{AssertEmpty}{i+1}$
\For{$x\in Q_i$}
\If{\Maybe{} $x>0$}
    \State\Restart
\EndIf
\EndFor
\EndProcedure\smallskip
\Procedure{\Name}{$i$} [$i=n{+}1$]
\If{\Maybe{} $\Reserve>0$}
    \State\Restart
\EndIf
\EndProcedure
\end{algorithmic}
\end{namedalgorithm}
\end{minipage}\hfill
\begin{minipage}[t]{0.47\textwidth}
\begin{namedalgorithm}[H]{AssertProper}
\caption{
}\label{alg:checkproper}
\begin{algorithmic}[1]
\Input $i\in\{1,...,n\}$
\Effect If $i$-proper or $i$-low, do nothing, else it may restart.
\Procedure{\Name}{$i$}
\State $\Call{AssertProper}{i-1}$
\For{$x\in\{x_i,y_i\}$}
	\If{\Maybe{} $x>0$}
	    \State\Restart
	\EndIf
	\State $\Call{Large}{\overline{x}}$
	\If{\Maybe{} $x>0$}
	    \State\Restart
	\EndIf
\EndFor
\EndProcedure
\end{algorithmic}
\end{namedalgorithm}
\end{minipage}
\end{figure}

\parag{Procedures \ref{alg:checkempty}, \ref{alg:checkproper}}
The procedure \ref{alg:checkempty} is supposed to determine whether a configuration is $i$-empty, which can easily be done by checking whether the relevant registers are nonempty.

Similarly, \ref{alg:checkproper} is used to ensure that the current configuration is not $i$-high. If it is, it may initiate a restart. We remark that calls to $\ref{alg:checkproper}(0)$ have no effect and can simply be omitted. 

\parag{Procedure \ref{alg:zero}}
This procedure implements a deterministic zero-check, as long as the register configuration is weakly $i$-proper. To ensure termination, \ref{alg:checkproper} is called within the loop.

\begin{figure}[t]
\begin{minipage}[t]{0.49\textwidth}
\begin{namedalgorithm}[H]{Zero}
\caption{Check whether a register is equal to $0$.}\label{alg:zero}
\begin{algorithmic}[1]
\Input $x\in\{x_i,\overline{x}_i,y_i,\overline{y}_i\}$
\Output whether $x=0$
\Procedure{\Name}{$x$}
\While{\True}
    \State \Call{AssertProper}{$i-1$}
    \If{\Maybe{} $x>0$}
        \State\Return\False
    \EndIf
    \If{$\Call{Large}{\overline{x}}$}
        \State\Return\True
    \EndIf
\EndWhile
\EndProcedure
\end{algorithmic}
\end{namedalgorithm}
\end{minipage}\hfill
\begin{minipage}[t]{0.49\textwidth}
\begin{namedalgorithm}[H]{IncrPair}
\caption{Decrement a two-digit, base $\beta:=N_i+1$ register}\label{alg:incrpair}%
\begin{algorithmic}[1]
\Input $x\in\{x_i,\overline{x}_i\}$, $y\in\{y_i,\overline{y}_i\}$
\Effect $\beta x+y\pmod{\beta^2}$ decreases by 1
\Procedure{\Name}{$x,y$}
\If{\Call{Zero}{$\overline{y}$}}
    \State \Swap{} $y,\overline{y}$\smallskip
    \If{\Call{Zero}{$\overline{x}$}}
        \State \Swap{} $x,\overline{x}$
    \Else
        {} $\overline{x}\mapsto x$
    \EndIf
\Else
    {} $\overline{y}\mapsto y$
\EndIf
\EndProcedure
\end{algorithmic}
\end{namedalgorithm}
\end{minipage}
\end{figure}

\parag{Procedure \ref{alg:incrpair}}
This is a helper procedure to increment the “virtual”, $N_{i+1}$-bounded counter simulated by $x$ and $y$. It works by first incrementing the second digit, i.e.\ $y$. If an overflow occurs, $x$ is incremented as well. It is also be used to decrement the counter, by running it on $\overline{x}$ and $\overline{y}$.

As we show later, \ref{alg:incrpair} is “reversible” under only the weak assumption that the configuration $C\in\N^Q$ is $i$-high. More precisely, $C,\ref{alg:incrpair}(x,y)\rightarrow C'$ implies $C',\ref{alg:incrpair}(\overline{x},\overline{y})\rightarrow C$. Using this, we can show that \ref{alg:large}, which calls \ref{alg:incrpair} in a loop, terminates.

\begin{namedalgorithm}[t]{Large}
\caption{Nondeterministically check whether a register is maximal.}\label{alg:large}%
\algrenewcommand\algorithmicindent{1.1em}
\begin{minipage}[t]{0.49\textwidth}
\begin{algorithmic}[1]
\Input $x\in\{x_i,\overline{x}_i,y_i,\overline{y}_i\},x\ne y$
\Output if $x\ge N_i$ return \True\ and swap units of $x-N_i$ and $\overline{x}$; or return \False
\smallskip
\Procedure{\Name}{$x$}\quad [for $i=1$]
\If{\Maybe{} $x>0$}
    \State $x\mapsto\overline{x}$
    \State \Swap\ $x,\overline{x}$
    \State\Return\True
\Else
    \State\Return\False
\EndIf
\EndProcedure
\end{algorithmic}
\end{minipage}\hfill
\begin{minipage}[t]{0.49\textwidth}
\begin{algorithmic}[1]
\makeatletter
\setcounter{ALG@line}{7}
\makeatother
\Procedure{\Name}{$x$}\quad [for $i>1$]
\If{$\neg\Call{Zero}{x_{i-1}}\vee\neg\Call{Zero}{y_{i-1}}$}
    \State\Restart
\EndIf
\While{\True}
    \State $\Call{CheckProper}{i-2}$
    \If{\Maybe{} $x>0$}
        \State $x\mapsto\overline{x}$
        \State\Call{IncrPair}{$x_{i-1},y_{i-1}$}
        \If{$\Call{Zero}{x_{i-1}}\wedge\Call{Zero}{y_{i-1}}$}
		    \State \Swap\ $x,\overline{x}$
           	\State\Return\True
        \EndIf
    \Else
        \If{$\Call{Zero}{x_{i-1}}\wedge\Call{Zero}{y_{i-1}}$}
            \State\Return\False
        \EndIf
        \If{\Maybe{} $\overline{x}>0$}
            \State $\overline{x}\mapsto x$
    		\State\Call{IncrPair}{$\overline{x}_{i-1},\overline{y}_{i-1}$}
        \EndIf
    \EndIf
\EndWhile
\EndProcedure
\end{algorithmic}
\end{minipage}
\end{namedalgorithm}

\parag{Procedure \ref{alg:large}}
This is the last of the subroutines, and the most involved one. The goal is to determine whether $x\ge N_i$, by using the registers of level $i-1$ to simulate a “virtual” $N_i$-bounded register. To ensure termination, we use a “random” walk, which nondeterministically moves either up or down. More concretely, at each step either $x$ is found nonempty, one unit is moved to $\overline{x}$ and the virtual register is incremented, or conversely $\overline{x}$ is nonempty, one unit moved to $x$, and the virtual register decremented. If the virtual register reaches $0$ from above, \ref{alg:large} had no effect and returns $\False$. Once the virtual register overflows, a total of $N_i$ units have been moved. These are put back into $x$ by swapping $x$ and $\overline{x}$ and $\True$ is returned.

As mentioned above, \ref{alg:incrpair} is reversible even under weak assumptions. This ensures that the random walk terminates, as it can always retrace its prior steps to go back to its starting point.

\begin{namedalgorithm}[t]{Main}
\caption{Decide whether there are at least $2\sum_iN_i$ agents.}\label{alg:main}
\begin{algorithmic}[1]
\Procedure{\Name}{}
\State $\OF :=\False$
\For{$i=1,...,n$}
    \While{$\neg\Call{Large}{\overline{x}_i}\vee\neg\Call{Large}{\overline{y}_i}$}
        \State\Call{AssertProper}{$i$}
        \State\Call{AssertEmpty}{$i+1$}
    \EndWhile
\EndFor
\State $\OF :=\True$
\While{\True}
    \State\Call{AssertProper}{$n$}
\EndWhile
\EndProcedure
\end{algorithmic}
\end{namedalgorithm}

\parag{Procedure \ref{alg:main}}
Finally, we put things together to arrive at the complete program. The implementation is very close to the steps described in Section~\ref{ssec:error-detection} in the simplified model, but instead of guessing an $i$ we iterate through the possibilities.

As mentioned before, \ref{alg:main} considers a small set of initial configurations “good” and may stabilise. 
The following lemma formalises this.

\begin{restatable}{lemma}{lemmain}\label{lem:main}
\ref{alg:main}, run on register configuration $C\in\N^Q$, can only restart or stabilise, and
\begin{enumerate}[label={(\alph*)}]
\item it may stabilise to $\False$ if $C$ is $j$-low and $(j+1)$-empty, for some $j\in\{1,...,n\}$,
\item it may stabilise to $\True$ if $C$ is $n$-proper, and
\item it always restarts otherwise.
\end{enumerate}
\end{restatable}

\section{Converting Population Programs into Protocols}\label{sec:conversion}
In the previous section we constructed succinct population programs for the threshold predicate. We now justify our model and prove that we can convert population programs into population protocols, keeping the number of states low. We do this in two steps; first we introduce population machines, which are a low-level representation of population programs, then we convert these into population protocols. This results in the following theorem:

\begin{restatable}{theorem}{thmconversion}\label{thm:conversion}
If a population program deciding $\varphi$ with size $n$ exists, then there is a population protocol deciding $\varphi'(x)\Leftrightarrow \varphi(x-i)\wedge x\ge i$ with $\O(n)$ states, for an $i\in\O(n)$.
\end{restatable}

Population machines are introduced in Section~\ref{ssec:formal-model}, they serve to provide a simplified model. Converting population programs into machines is straightforward and uses standard techniques, similar to how one would convert a structured program to use only goto-statements. We will describe this in Section~\ref{ssec:high-to-low}. The conversion to population protocols is finally described in Section~\ref{ssec:low-to-pp}.
Here, we only highlight the key ideas of the conversion. The full details can be found in Appendix~\ref{app:detailedconversion}.

\subsection{Formal Model}\label{ssec:formal-model}

\begin{definition}
A \emph{population machine} is a tuple $\Machine=(Q,\Flags,\Flagdom,\Inst )$, where $Q$ is a finite set of \emph{registers}, $\Flags$ a finite set of \emph{pointers}, $\Flagdom=(\Flagdom_i)_{i\in\Flags}$ a list of \emph{pointer domains}, each of which is a nonempty finite set, and $\Inst=(\Inst_1,...,\Inst_L)$ is a sequence of \emph{instructions}, with $L\in\N$. Additionally, $\OF,\CF,\IP\in\Flags$, $\Flagdom_{\OF}=\Flagdom_{\CF}=\{\False,\True\}$ and $\Flagdom_{\IP}=\{1,..,L\}$. For $x\in Q\cup\{\square\}$ we also require $\Virt{x}\in\Flags$, and $x\in\Flagdom_{\Virt{x}}\subseteq Q$. The \emph{size} of $\Machine$ is $|Q|+|\Flags|+\sum_{X\in\Flags}|\Flagdom_X|+|\Inst|$.

Let $x,y\in Q$, $x\ne y$, $X,Y\in\Flags$, $i\in\{1,...,L\}$ and $f:\Flagdom_Y\rightarrow\Flagdom_X$. There are three types of instructions: $\Inst_i=(x\mapsto y)$, $\Inst_i=(\Maybe\ x>0)$, or $\Inst_i=(X:=f(Y))$.
\end{definition}

A population machine has a number of registers, as usual, and a number of pointers. While each register can take any value in $\N$, a pointer is associated with a finite set of values it may assume. There are three special pointers: the output flag $\OF$, which we have already seen in population programs and is used to indicate the result of the computation, the condition flag $\CF$ used to implement branches, and the instruction pointer $\IP$, storing the index of the next instruction to execute. To implement swap instructions we use a register map; the pointer $\Virt{x}$, for a register $x\in Q$, stores the register $x$ is actually referring to. ($\Virt{\square}$ is a temporary pointer for swapping.) The model allows for arbitrary additional pointers, we will use a one per procedure to store the return address.

There are only three kinds of instructions: $(x\mapsto y)$ and $(\Maybe\ x>0)$ are present in population programs as well and have the same meaning here. (With the slight caveat that $x$ and $y$ are first transformed according to the register map. The instructions do not operate on the actual registers $x,y$, but on the registers pointed to by $V_x$ and $V_y$.) The third, $(X:=f(Y))$ is a general-purpose instruction for pointers. It can change $\IP$ and will be used to implement control flow constructs.

A precise definition of the semantics can be found in Appendix~\ref{app:machinesemantics}.

\subsection{From Population Programs to Machines}\label{ssec:high-to-low}
Population machines do not have high-level constructs such as loops or procedures, but these can be implemented as macros using standard techniques. We show only an example here, a detailed description of the conversion can be found in Appendix~\ref{app:detailedmachines}.

\begin{figure}[H]
\centering
\hspace*{6mm}
\begin{minipage}{0.36\textwidth}
\begin{algorithmic}
\Procedure{Main}{}
\color{nicered}\While{\color{niceblue}$\Maybe\ x>0$\color{nicered}}
    \color{nicepurple}\State $x\mapsto y$
\color{nicered}\EndWhile
\State \textcolor{nicegreen}{\Swap\ $x,y$}
\EndProcedure
\end{algorithmic}
\end{minipage}\scalebox{2}{$\rightsquigarrow$}\hspace{6mm}
\begin{minipage}{0.4\textwidth}
\begin{algorithmic}[1]
\color{niceblue}\State \Maybe\ $x>0$
\color{nicered}\State $\IP:=\Big\{\begin{array}{ll}
5&\text{ if $\CF$}\\
3&\text{ else}
\end{array}$
\color{nicepurple}\State $x\mapsto y$
\color{nicered}\State $\IP:=1$
\color{nicegreen}\State $V_\square:=V_x$
\State $V_x:=V_y$
\State $V_y:=V_\square$
\end{algorithmic}
\end{minipage}
\caption{Conversion to a population machine.}\label{fig:convert-to-machine}%
\end{figure}

Control-flow, i.e.\ \textbf{if}, \textbf{while} and procedure calls are implemented via direct assignment to $\IP$, the instruction pointer, as in lines 2 and 4 above. The statements $(\Maybe\ x>0)$ and $(x\mapsto y)$ are translated one-to-one, but note that in the population machine their operands are first translated via the register map. For example, $(\Maybe\ x>0)$ in line 1 checks whether the register pointed to by $V_x$ is nonzero. Correspondingly, \Swap\ statements result in direct modifications to the register map: lines 5-7 swap the pointers $V_x$ and $V_y$ (and leave the registers they point to unchanged).

\subsection{Conversion to Population Protocols}\label{ssec:low-to-pp}
In this section, we only present a simplified version of our construction. In particular, we make use of multiway transitions to have more than two agents interact at a time. Our actual construction, described in Appendix~\ref{app:detailedprotocols}, avoids them and the associated overhead.

Let $\Machine=(Q,\Flags,\Flagdom,\Inst)$ denote a population machine. To convert this into a population protocol, we use two types of agents: \emph{register agents} to store the values of the registers, and \emph{pointer agents} to store the pointers. For a register we have many identical agents, and the value of the register corresponds to the total number of those agents. They use states $Q$. For each pointer we use a unique agent, storing the value of the pointer in its state; they use states $\{X^v:X\in F,v\in\Flagdom_X\}$.

Let $X_1,...,X_{\Abs{\Flags}}$ denote some enumeration of $\Flags$ with $X_{\Abs{\Flags}}=\IP$, and let $v_i$ denote the initial value of $X_i$. We use $X_1$ as initial state of the protocol. To goal is to have a unique agent for each pointer, so we implement a simple leader election. We use $\ast$ as wildcard.
\[\begin{array}{rlrl}
X_i^*, X_i^*&\mapsto X_i^{v_i}, X_{i+1}^{v_{i+1}}&\qquad
\IP^*, \IP^*&\mapsto X_1^{v_1}, x
\end{array}\]
with $i\in\{1,...,\Abs{\Flags}-1\}$. If two agents store the value of a single pointer, they eventually meet and one of them is moved to another state. When this happens, the computation is restarted ― but note that the values of the registers are not reset. Eventually, the protocol will thus reach a configuration with exactly one agent in $X_i^{v_i}$, for each $i$, and the remaining agents in $Q$.

Starting from this configuration, the instructions can be executed. We illustrate the mapping from instructions to transitions in the following example:

\begin{figure}[H]
\centering
\begin{minipage}{0.24\textwidth}
\begin{algorithmic}[1]
\color{nicepurple}\State $x\mapsto y$
\color{niceblue}\State \Maybe\ $x>0$
\color{nicered}\State $\IP:=\Big\{\begin{array}{ll}
1&\text{ if $\CF$}\\
4&\text{ else}
\end{array}$
\color{nicegreen}\State $\OF:=\neg\CF$
\end{algorithmic}
\end{minipage}\scalebox{2}{$\rightsquigarrow$}\hspace{6mm}
\begin{minipage}{0.6\textwidth}
\[\hspace*{-10mm}\begin{array}{lll}
\color{nicepurple}\IP^1, V_x^v, V_y^w, v &\color{nicepurple}\mapsto \IP^2, V_x^v, V_y^w, w&\quad\text{for }v,w\in Q\\
\color{niceblue}\IP^2, \CF^*, V_x^v, v&\color{niceblue}\mapsto \IP^3, \CF^\TrueNC, V_x^v, v&\quad\text{for }v\in Q\\
\color{niceblue}\IP^2, \CF^*, V_x^v, w&\color{niceblue}\mapsto \IP^3, \CF^\FalseNC, V_x^v, w&\quad\text{for }w\ne v\\
\color{nicered}\IP^3, \CF^\TrueNC&\color{nicered}\mapsto \IP^1, \CF^\TrueNC&\\
\color{nicered}\IP^3, \CF^\FalseNC&\color{nicered}\mapsto \IP^4, \CF^\FalseNC&\\
\color{nicegreen}\IP^4, \OF^*,\CF^\TrueNC&\color{nicegreen}\mapsto \IP^5, \OF^\FalseNC,\CF^\TrueNC&\\
\color{nicegreen}\IP^4, \OF^*,\CF^\FalseNC&\color{nicegreen}\mapsto \IP^5, \OF^\TrueNC,\CF^\FalseNC&
\end{array}\]
\end{minipage}
\caption{Converting instructions into transitions.}\label{fig:convert-to-protocol}%
\end{figure}

For example, in line~1 we want to move one agent from $x$ to $y$ and set the instruction pointer to $2$ (from $1$). Recall that the registers map to states of the population protocol via the register map, stored in pointers $V_x$, where $x\in Q$ is a register. We thus have the following agents initiating the transition:
\begin{itemize}
\item $\IP^1$; the agents storing the instruction pointer currently stores the value $1$,
\item $V_x^v$; the register $x\in Q$ is currently mapped to state $v\in Q$,
\item $v$; an agent in state $v$, i.e.\ representing one unit in register $x$,
\item $V_y^w$; register $y$ is mapped to state $w$.
\end{itemize}
The transition then moves $v$ to state $w$, and increments the instruction pointer.

The above protocol does not come to a consensus. For this to happen, we use a standard output broadcast: we add a single bit to all states. In this bit an agent stores its current opinion. When any agent meets the pointer agent of the output flag $\OF$, the former will assume the opinion of the latter. Eventually, the value of the output flag has stabilised and will propagate throughout the entire population, at which point a consensus has formed.

\section{Robustness of Threshold Protocols}\label{sec:robustness}
A major motivation behind the construction of succinct protocols for threshold predicates is the application to chemical reactions. In this, as in other environments, computations must be able to deal with errors. Prior research has considered \emph{self-stabilising} protocols~\cite{AngluinAFJ05,CaiIW12,BurmanCCDNSX21}. Such a protocol must converge to a desired output regardless of the input configuration. However, it is easy to see that no population protocol for e.g.\ a threshold predicate can be self-stabilising (and prior research has thus focused on investigating extensions of the population protocol model).

In our definition of population programs, the program cannot rely on any guarantees about its input configuration, so they are self-stabilising by definition. However, when we convert to population protocols, we retain only a slightly weaker property, defined as follows:

\begin{definition}
Let $\Prot=(Q,\delta,I,O)$ denote a population protocol deciding $\varphi$ with $\Abs{I}=1$. We say that $\Prot$ is \emph{almost self-stabilising}, if every fair run starting at a configuration $C\in\N^Q$ with $C(I)\ge\Abs{Q}$ stabilises to $\varphi(\Abs{C})$.
\end{definition}

So the initial configuration can be almost arbitrary, but it must contain a small number of agents in the initial state. In many contexts, this is a mild restriction. In a chemical reaction, for example, the number of agents (i.e.\ the number of molecules) is many orders of magnitude larger than the number of states (i.e.\ the number of species of molecules).

In particular, this is also much stronger than any prior construction. All known protocols for threshold predicates are 1-aware~\cite{BlondinEJ18}, and can thus be made to accept by placing a single agent in an accepting state.

\restateMainTwo*
\begin{proof}
The proof is exactly analogous to the proof of Proposition~\ref{pro:low-to-pp}, since Lemma~\ref{lem:electworks} works for any configuration with at least $\Abs{\Flags}$ agents in the initial state, and $\Abs{\Flags}\le\Abs{Q^*}$.
\end{proof}

\section{Conclusions}
We have shown an $\O(\log\log n)$ upper bound on the state complexity of threshold predicates for leaderless population protocols, closing the last remaining gap. Our result is based on a new model, population programs, which enable the specification of leaderless population protocols using structured programs.

As defined, our model of population programs can only decide unary predicates and it seems impossible to decide even quite simple remainder predicates (e.g.\ “is the total number of agents even”). Is this a fundamental limitation, or simply a shortcoming of our specific choices? We tend towards the latter, and hope that other very succinct constructions for leaderless population protocols can make use of a similar approach.

Our construction is almost self-stabilising, which shows that it is possible to construct protocols that are quite robust against \emph{addition} of agents in arbitrary states. A natural next step would be to investigate the \emph{removal} of agents: can a protocol provide guarantees in the case that a small number of agents disappear during the computation?

Threshold predicates can be considered the most important family for the study of space complexity, as they are the simplest way of encoding a number into the protocol. The precise space complexity of other classes of predicates, however, is still mostly open. The existing results generalise somewhat; the construction presented in this paper, for example, can also be used to decide $\varphi(x)\Leftrightarrow x=k$ for $k\ge 2^{2^n}$ with $\O(n)$ states. As mentioned, there also exist succinct constructions for arbitrary predicates, but --- to the extent of our knowledge --- it is still open whether, for example, $\varphi(x)\Leftrightarrow x=0\pmod k$ can be decided for $k\ge2^{2^n}$, both with and without leaders.

\bibliography{references}

\appendix

\appendix
\section{Proofs of Section~\ref{sec:details}}\label{app:details}
In this section, we prove correctness of our construction of the population programs in Section~\ref{sec:details}. First, we introduce the necessary formal definitions to precisely state the guarantees of each procedure.
\medskip

\parag{Definitions}
Let $C\in\N^Q$ and $i\in\{1,...,n\}$. We say that $C$ is
\begin{itemize}
\item \emph{$i$-proper}, if $C(x_j)=C(y_j)=0$ and $C(\overline{x}_j)=C(\overline{y}_j)=N_j$ for $j\in\{1,...,i\}$
\item \emph{weakly $i$-proper}, if $C$ is $(i-1)$-proper and $C(x)+C(\overline{x})=N_i$ for $x\in\{x_i,y_i\}$
\item \emph{$i$-low}, if $C$ is $(i-1)$-proper, not $i$-proper, and $C(x)=0$ and $C(\overline{x})\le N_i$ for all $x\in\{x_i,y_i\}$
\item \emph{$i$-high}, if $C$ is $(i-1)$-proper, not $i$-proper, and $C(x)+C(\overline{x})\ge N_i$ for all $x\in\{x_i,y_i\}$
\item \emph{$i$-empty}, if $C(x)=0$ for all $x\in Q_i\cup...\cup Q_n\cup\{\Reserve\}$
\end{itemize}

A procedure $f$ is \emph{$i$-robust} if for all $i$-high $C$ we have $C,f\not\rightarrow\bot$, and $C,f\rightarrow C',b$ (or $C,f\rightarrow C'$) implies that $C'$ is $i$-high as well. Note that $C,f\rightarrow\Restart$ is allowed. Finally, $f$ is \emph{robust} if it is $i$-robust for all $i\in\{1,...,n\}$.

We set $\Counter_{x,y}(C):=C(x)\cdot(N_i+1)+C(y)$ to the value of the two-digit, base $N_i+1$ counter using $x$ and $y$ as digits, where $C\in\N^Q$, $i\in\{1,...,n\}$ and $x\in\{x_i,\overline{x}_i\},y\in\{y_i,\overline{y}_i\}$.

We sometimes write $\{x:\alpha\}$, where $\alpha$ is independent of $x$. This denotes either $\{x\}$, if $\alpha$, or $\emptyset$ otherwise.

\subsection{\ref{alg:checkempty}}
\begin{restatable}{lemma}{lemcheckempty}\label{lem:checkempty}
Let $C\in\N^Q$, $i\in\{1,...,n{+}1\}$. Then $\Post(C,\ref{alg:checkempty}(i))=\{C\}\cup S$, where $S=\emptyset$ if $C$ is $i$-empty and $S=\{\Restart\}$ otherwise. Moreover, $\ref{alg:checkempty}(i)$ is robust.
\end{restatable}
\begin{proof}
Clearly, \ref{alg:checkempty} cannot affect any register, and restarts only if one of the registers $Q_i\cup...\cup Q_n\cup\{\Reserve\}$ is nonzero. Robustness follows immediately.
\end{proof}

\subsection{\ref{alg:checkproper},~\ref{alg:zero},~\ref{alg:incrpair}, and~\ref{alg:large}}
The procedures \ref{alg:checkproper}, \ref{alg:zero}, \ref{alg:incrpair}, and \ref{alg:large} are instantiated for each level and call each other. Population programs allow only acyclic procedure calls, so the correctness proofs can proceed inductively and rely on the correctness of all called procedures. To be formally precise, we must note that the proofs of the following lemmata do not prove the associated lemma independently of the others. They only prove part of the induction step, and only if all proofs work do the statements of the lemmata follow. We therefore label them “proof fragments”.

\begin{restatable}{lemma}{lemcheckproper}\label{lem:checkproper}
Let $C\in\N^Q$, $i\in\{1,...,n\}$. Then 
\begin{enumerate}[label={(\alph*)}]
\item $\Post(C,\ref{alg:checkproper}(i))=\{C\}$ if $C$ is $i$-proper or $i$-low,
\item $C,\ref{alg:checkproper}(i)\rightarrow\Restart$ if $C$ is $j$-high, for some $j\in\{1,...,i\}$,
\item $C,\ref{alg:checkproper}(i)\rightarrow\Restart$ if $C$ is $(i-1)$-proper and $C(x)>0\vee C(\overline{x})>N_i$, for some $x\in\{x_i,y_i\}$, and
\item $\ref{alg:checkproper}(i)$ is robust.
\end{enumerate}
\end{restatable}
\begin{proof}[Proof (fragment)]
\textbf{(a)} By induction, the recursive call in line 2 must return with $C$. As $C$ is weakly $i$-proper, line 6 has no effect (Lemma~\ref{lem:large}a) and neither line 5 nor line 8 is executed.

\textbf{(b)} The case $j<i$ is covered inductively, otherwise it follows directly from (c).

\textbf{(c)} If $C(x)>0$, line 5 may execute a restart. If $C(\overline{x})>N_i$, we use Lemma~\ref{lem:large}b to derive that $x$ may be nonzero at line 7. If $x=y_i$, we must also note that Lemma~\ref{lem:large}b ensures that the first iteration of the for-loop either restarts or terminates without affecting $x$ and $\overline{x}$.

\textbf{(d)} Let $C$ be a $j$-high configuration, for some $j$. If $j>i$ then we need only invoke property (a). Otherwise, we use that \ref{alg:checkproper} and \ref{alg:large} are robust (Lemma~\ref{lem:large}c and induction), so their execution terminates and does not affect whether the configuration is $j$-high.
\end{proof}

\begin{restatable}{lemma}{lemzero}\label{lem:zero}
Let $i\in\{1,...,n\}$, $x\in\{x_i,\overline{x}_i,y_i,\overline{y}_i\}$, $C,C'\in\N^Q$. Then
\begin{enumerate}[label={(\alph*)}]
\item $\Post(C,\ref{alg:zero}(x))=\{(C,C(x)=0)\}$ if $C$ is weakly $i$-proper,
\item $\Post(C,\ref{alg:zero}(x))=\{(C,\False):C(x)>0\}\cup\{(C',\True):C(\overline{x})\ge N_i)\}$ if $C$ is $(i-1)$-proper and $C(x)+C(\overline{x})\ge N_i$, where $C'(\overline{x})=C(x)+N_i$, $C'(x)=C(\overline{x})-N_i$ and $C'(z)=C(z)$ for $z\notin\{x,\overline{x}\}$
\item $C,\ref{alg:zero}(x)\rightarrow C',\False$ implies $C'(x)>0$, for all $C'$, and
\item $\ref{alg:zero}(x)$ is robust.
\end{enumerate}
\end{restatable}
\begin{proof}[Proof (fragment)]
\textbf{(a)} This follows immediately from (b): if $C(x)+C(\overline{x})=N_i$ then $C(x)=0$ is equivalent to $C(\overline{x})\ge N_i$, and $C'=C$ (assuming $C(\overline{x})\ge N_i$).

\textbf{(b)} As $C$ is $(i-1)$-proper, the call to \ref{alg:checkproper} has no effect (Lemma~\ref{lem:checkproper}a). Further, \ref{alg:large} has no effect as long as it returns \False\ (Lemma~\ref{lem:large}b). Hence, for all iterations of the loop, registers start in $C$. Line 5, therefore, may execute iff $C(x)>0$. Again due to Lemma~\ref{lem:large}b, line 7 can execute iff $C(\overline{x})\ge N_i$, and if so, registers are according to $C'$. Finally, either $C(x)>0$ or $C(\overline{x})\ge N_i$ holds, so eventually line 5 or line 7 will return the correct result due to fairness and the procedure terminates.

\textbf{(c)} This follows from the observation that \False\ can only be returned in line 5.

\textbf{(d)} Let $C$ be $j$-high. If $j>i$ we can invoke property (a). For $j=i$ we use (b), noting that $C'$ is still $i$-high. Otherwise, we use that \ref{alg:checkproper} and \ref{alg:large} are robust and do not affect whether the register configuration is $j$-high. Finally, we know that line~3 is eventually going to restart (Lemma~\ref{lem:checkproper}b and fairness), so the loop cannot repeat infinitely often.
\end{proof}

We want to highlight property (b) of the following lemma; it states that \ref{alg:incrpair} is “reversible” in some sense, under only the weak assumption that the configuration is $i$-high (or $i$\nobreakdash-proper). We need this property later to show that \ref{alg:large} is robust.

Regarding property (c) we remark that, contrary to the other procedures, \ref{alg:incrpair} is not $j$-robust for all $j$, but only $j\le i$. This is simply due to the fact that it is designed to change the value of level $i$ registers; if executed on an $i$-proper configuration it results only in a weakly $i$-proper configuration.

\begin{restatable}{lemma}{lemincrpair}\label{lem:incrpair}
Let $i\in\{1,...,n\}$, $x\in\{x_i,\overline{x}_i\},y\in\{y_i,\overline{y}_i\}$, $C,C'\in\N^Q$. Then
\begin{enumerate}[label={(\alph*)}]
\item $\Post(C,\ref{alg:incrpair}(x,y))=\{C'\}$ if $C$ is weakly $i$-proper, where $C'$ is the unique weakly $i$-proper multiset with $\Counter_{x,y}(C')=\Counter_{x,y}(C)+1\pmod{N_{i+1}}$ and $C'(w)=C(w)$ for $w\notin\{x_i,\overline{x}_i,y_i,\overline{y}_i\}$,
\item $C,\ref{alg:incrpair}(x,y)\rightarrow C'$ implies both $C',\ref{alg:incrpair}(\overline{x},\overline{y})\rightarrow C$ and $C'(z)=C(z)$ for $z\notin Q_i$, if $C$ is $(i-1)$-proper and $C(w)+C(\overline{w})\ge N_i$ for $w\in\{x_i,y_i\}$, and
\item $\ref{alg:incrpair}(x,y)$ is $j$-robust, for $j\le i$.
\end{enumerate}
\end{restatable}
\begin{proof}[Proof (fragment)]
\textbf{(a)} If $C$ is weakly $i$-proper, the calls to \ref{alg:zero} work deterministically and the registers $x$ and $y$ are adjusted according to the specification: line 2 checks whether $y$ (the least significant digit) is $N_i$. If not, it is incremented. Otherwise, it overflows; $y$ is set to $0$ and $x$ is incremented, checking whether it overflows as well. Finally, note $N_{i+1}=(N_i+1)^2$.

\textbf{(b)} The property $C'(z)=C(z)$ for $z\notin Q_i$ follows immediate from Lemma~\ref{lem:zero}b. In particular, lines 4-6 only affect the values of $x$ and $\overline{x}$, while lines 2,3 and 7 only affect $y$ and $\overline{y}$. We now consider executing \ref{alg:incrpair} twice, first with arguments $x,y$, then with $\overline{x},\overline{y}$. We start with registers $C$, and argue that it is possible for the second execution to take the same branches (in lines 2 and 4) as the first. Afterwards we derive that the registers again have values $C$.

Consider line 2. If the branch is not taken, \ref{alg:zero} had no effect. After $\overline{y}\mapsto y$ in line 7, clearly $C'(y)>0$. In the second execution, line 2 runs $\ref{alg:zero}(y)$ (recall that the second execution has different arguments). This may now return \False\ and the same branch is taken.

If the branch in line 2 is taken, after line 3 registers $y,\overline{y}$ have been changed. More precisely, $N_i$ units have been moved from $y$ to $\overline{y}$. Lines 4-6 do not affect $y,\overline{y}$, so $C'(\overline{y})\ge N_i$. In the second execution, the call $\ref{alg:zero}(y)$ may then return \True.

The argument for the branch in line 4 is analogous. Finally, we argue that, if the same branches are taken, the second execution undoes the changes of the first. Briefly, if the branch in line 2 is not taken, only line 7 changes any registers. Clearly, executing $\overline{y}\mapsto y$ and then $y\mapsto\overline{y}$ has no effect. If it is taken, the combined effect of lines 2 and 3 is moving $N_i$ units from $y$ to $\overline{y}$, which are then moved back in the second execution. Again the situation for lines 4-6 is analogous.

\textbf{(c)} Let $C$ be $j$-high, for $j\le i$. As \ref{alg:zero} is robust, it does not affect whether the register configuration is $j$-high and either terminates or restarts. Lines 3,5,6 and 7, if executed, also do not affect $j$-highness. Finally, there is no loop and Lemma~\ref{lem:zero}c implies that lines 6 and 7 cannot hang, so \ref{alg:incrpair} either terminates or restarts.
\end{proof}

\begin{restatable}{lemma}{lemlarge}\label{lem:large}
Let $i\in\{1,...,n\}$, $x\in\{x_i,\overline{x}_i,y_i,\overline{y}_i\}$, and $C\in\N^Q$. Then
\begin{enumerate}[label={(\alph*)}]
\item $\Post(C,\ref{alg:large}(x))=\{(C,\False),(C,C(x)\ge N_i)\}$ if $C$ is weakly $i$-proper,
\item $\Post(C,\ref{alg:large}(x))=\{(C,\False)\}\cup\{(C',\True):C(x)\ge N_i\}$ if $C$ is $(i-1)$-proper, with $C'(x)=C(\overline{x})+N_i$, $C'(\overline{x})=C(x)-N_i$ and $C'(z)=C(z)$ for $z\notin\{x,\overline{x}\}$, and
\item $\ref{alg:large}(x)$ is robust.
\end{enumerate}
\end{restatable}
\begin{proof}[Proof (fragment)]
\textbf{(a)} Follows directly from (b); if $C(x)\ge N_i$ and $C$ is weakly $i$-proper, then $C(x)=N_i$ and $C(\overline{x})=0$, which implies $C'=C$.

\textbf{(b)} The case $i=1$ is trivial. Assume $i>1$. The registers will remain in a weakly $(i-1)$-proper configuration; lines 14, 17 and 23 do not affect this, and neither do the calls to \ref{alg:incrpair} (Lemma~\ref{lem:incrpair}a), to \ref{alg:checkproper} (Lemma~\ref{lem:checkproper}a), nor to \ref{alg:zero} (Lemma~\ref{lem:zero}a). As the registers are weakly $(i-1)$-proper, the calls to \ref{alg:zero} work as intended and deterministically check whether the register is zero (again, Lemma~\ref{lem:zero}a). In particular, using $C(x_{i-1})=C(y_{i-1})=0$ we find that line 10 cannot execute. Additionally, since the registers remain weakly $(i-1)$-proper and thus $(i-2)$-proper, line 12 has no effect (Lemma~\ref{lem:checkproper}a).

We consider the register simulated by \ref{alg:incrpair}; for convenience we introduce the shorthand $\Counter:=\Counter_{x_{i-1},y_{i-1}}$. As $C$ was $(i-1)$-proper, $\Counter(C)=0$. This counter is only modified by the calls to \ref{alg:incrpair}, as specified by Lemma~\ref{lem:incrpair}a. Line 15 increments the counter, and line 24 decrements it. Line 15 may overflow the counter, but then the branch in line 16 will immediately be taken. Line 24 can only execute of the check in line 20 fails, so it cannot underflow the counter.

As the counter neither over- nor underflows, for any register configuration $C^*$ the procedure reaches at the beginning of the loop in line 12, $\Counter(C^*)$ correspond to units moved from $x$ to $\overline{x}$ via lines 14 and 23.

We now show $C,\ref{alg:large}(x)\rightarrow C,\False$ and, if $C(x)\ge N_i$, $C,\ref{alg:large}(x)\rightarrow C',\True$. For the former, we even show the stronger property that $C,\False$ can be returned from any iteration of the loop. Let $C^*$ denote some configuration reached at line 12. From now on, we never take the branch in line 13. If $\Counter(C^*)=0$, then we claim $C^*(z)=C(z)$ for all $z$. If $z$ has level at most $i-2$ this follows from $C^*$ being weakly $(i-1)$-proper. If $z$ has level $i-1$, we use $\Counter(C^*)=0=\Counter(C)$. For $z$ at level $i$ or above, note that only registers $x$ and $\overline{x}$ can be modified by the procedure, but $\Counter(C^*)=0$ ensures that no units have moved between them. Using $C^*=C$ we now see that the branch in line 20 can be taken and we return \False\ with registers $C$.

If $\Counter(C^*)>0$, then the branch in line 20 cannot be taken. Using $C^*(\overline{x})\ge\Counter(C^*)>0$, we can take the branch in line 22. In the next iteration of the loop we have decreased $\Counter$ by one; the property then follows from induction. We remark that this also shows that the procedure always terminates.

We now prove $C,\ref{alg:large}(x)\rightarrow C',\True$, assuming $C(x)\ge N_i$. Here, it is possible to take the branch in line 13 $N_i$ times and we do. Afterwards, the counter overflows and line 18 returns \True. As before, the only registers that may have changed relative to $C$ are $x$ and $\overline{x}$. We moved $N_i$ units from $x$ to $\overline{x}$, swapping them then results in $C'$.

Finally, we need to show that the above two cases cover all possibilities. We already argued that the procedure always terminates and no restart can occur. If we return in line 18, the counter was overflowed and $N_i$ units have been moved, resulting in $C'$. If we return in line 21, changes to $x$ and $\overline{x}$ have cancelled out, and we are in $C$.

\textbf{(c)} Let $C$ be a $j$-high configuration, for some $j$. If $j\ge i$ we need only refer to (b), noting that $C'$ is still $j$-high. For $j<i$ we can rely on \ref{alg:checkproper}, \ref{alg:zero} and \ref{alg:incrpair} being $j$-robust (lemmata~\ref{lem:checkproper}d,~\ref{lem:zero}d and~\ref{lem:incrpair}c). In particular, they do not affect whether the register configuration is $j$-high. Neither do lines 14, 17 or 23, so the registers stay $j$-high. Additionally, this yields that the calls to these procedures terminate or restart.

It remains to argue that the loop terminates. If $j\le i-2$ this is ensured by \ref{alg:checkproper} (Lemma~\ref{lem:checkproper}b), so we are left with $j=i-1$. In this case the call to \ref{alg:checkproper} in line 12 has no effect and we shall ignore it. Further, note that the calls to \ref{alg:zero} and \ref{alg:incrpair} can only change a register $z$ if $z$ or $\overline{z}$ is one of their arguments (lemmata~\ref{lem:zero}b and~\ref{lem:incrpair}b).

Let $\Confs$ denote the set of $j$-high configurations. For $D,D'\in\Confs$ we write $D\Rel D'$ if one iteration of the loop (i.e.\ executing lines 12-24 in sequence), starting with registers according to $D$, may end with registers in $D'$ (without returning). We now claim that $\Rel$ is symmetric. To see that this claim suffices, let $C^*$ denote the register configuration after line 9. Using Lemma~\ref{lem:zero}b, $C^*(\overline{x}_{i-1}),C^*(\overline{y}_{i-1})\ge N_i$ hold. Our claim then implies that the loop can go back to $C^*$ after any number of iterations. Eventually, it will do so due to fairness. Then, it may take the else branch in line 19. Using Lemma~\ref{lem:zero}b again, line 21 may execute and the procedure returns.

We now show the claim. Fix $D,D'$ with $D\Rel D'$. There are now two cases: either $D'$ results from $D$ by executing lines 14-16, or lines 20-24. We now need to argue that $D$ may result if the loop starts with $D'$. Consider the first case. The else branch in line 19 may always be taken, so it suffices that lines 20-24 may undo the effects of lines 14-16 from earlier. Due to line 14, $D'(\overline{x})>0$, and the branch in line 22 may be taken. Using Lemma~\ref{lem:zero}b, lines 16 and 20 may cancel out, Lemma~\ref{lem:incrpair}b implies that lines 15 and 24 may cancel, and lines 14 and 23 undo each other as well.

The argument for the second case is analogous. There, line 23 ensures that we can subsequently take the branch in line 13, and the lines cancel in the same manner.
\end{proof}

\subsection{\ref{alg:main}}
\lemmain*
\begin{proof}
The output register $\OF $ is only changed by lines 2 and 7. (This can easily be checked syntactically; no called procedure uses $\OF $.) So either the execution restarts; or one of the two loops in lines 4 and 8 does not terminate and the computation stabilises.

Before moving to claims (a-c), we argue that, if $C$ is $i$-proper, the $i$-th iteration of the for-loop in line 3 may terminates without effect, otherwise it restarts. Here, we use Lemma~\ref{lem:large}a to derive that line 4 has no effect and that the loop condition may be $\False$; due to fairness the loop terminates eventually. Line 5 has no effect as well (Lemma~\ref{lem:checkproper}a), and line 6 either restarts or does nothing (Lemma~\ref{lem:checkempty}).

\textbf{(a)} $C$ is $(j-1)$-proper, so, as argued above, iterations $i\in\{1,...,j-1\}$ of the for-loop may terminate without changing a register, and they restart otherwise. In iteration $i=j$ the while-loop in line 4 cannot terminate, and lines 5-6 have no effect and cannot initiate a restart, so the computation stabilises to \False.

\textbf{(b)} Now all $n$ iterations of the for-loop in line 3 may terminate without effect (or restart, otherwise). If they do, we enter the second while-loop, in line 8, and stabilise to \True.

\textbf{(c)} Let $j\in\{1,...,n\}$ be maximal s.t.\ $C$ is $(j-1)$-proper. (Such a $j$ always exists.) As argued before, the first $j-1$ iterations of the for-loop cannot change any registers. In iteration $i=j$, we have that \ref{alg:checkempty} and \ref{alg:large} always terminate (lemmata~\ref{lem:checkempty} and~\ref{lem:large}b).

There are the following cases.

Case 1, $C$ is $j$-low and not $(j+1)$-empty. As we have argued for (a), in iteration $i=j$ the loop in line 4 cannot terminate, so eventually line 6 will initiate a restart (Lemma~\ref{lem:checkempty}).

Case 2, $C(\overline{x})<N_j$ for some $x\in\{x_j,y_j\}$. We may assume that $C$ is not $j$-low, since we have already covered that possibility in (a) and Case 1. Hence, we have $C(y)>0$ or $C(\overline{y})>N_i$ for some $y\in\{x_i,y_i\}$. In iteration $i=j$ thus \ref{alg:checkproper} either terminates or it may initiate a restart (Lemma~\ref{lem:checkproper}c). However, $\ref{alg:large}(\overline{x})$ will always return \False, so the loop in line 4 repeats infinitely often. Due to fairness, a restart must eventually happen.

Case 3, $C$ is $j$-high. As \ref{alg:checkempty}, \ref{alg:large}, and \ref{alg:checkproper} are robust (lemmata~\ref{lem:checkempty},~\ref{lem:checkproper}d and~\ref{lem:large}c), they terminate or restart and the register configuration will remain $j$-high. Assuming that no restart occurs, we know that the subsequent computation would execute $\ref{alg:checkproper}(k)$ infinitely often, for some $k\ge j$. (This occurs either in line 5, or line 9.) However, Lemma~\ref{lem:checkproper}b guarantees that these calls may restart, so a restart will happen eventually due to fairness.

Note that the above case distinction is exhaustive, as $C$ cannot be $j$-proper (either $j$ would not be maximal, or $C$ would be $n$-proper).
\end{proof}

\subsection{Proof of Theorem~\ref{thm:exists-program}}
\restateExistsProgram*
\begin{proof}
We define $k:=2\sum_{i=1}^nN_i$. (Recall that $N_{i+1}=(N_i+1)^2$ and $N_1=1$, implying $k\ge 2^{2^n}$.)

Let $m\in\N$ and let $\Confs:=\{C\in\N^Q:|C|=m\}$ denote the configurations where registers sum to $i$. It suffices to show that $\Confs$ contains a “good” configuration; i.e.\ an $n$-proper configuration iff $m\ge k$, or a $j$-low and $(j+1)$-empty configuration for some $j\in\{1,..,n\}$ iff $m<k$. If these hold, Lemma~\ref{lem:main} guarantees that every run starting with another kind of configuration eventually restarts. By fairness, at some point the computation restarts with a good configuration and stabilises to the correct output.

It remains to argue that the above claim holds. If $m\ge k$, we note that superfluous units can be left in register $\Reserve$, keeping the configuration $n$-proper; conversely, any $n$-proper configuration $C$ clearly has $|C|\ge k$. Otherwise, a good configuration can have at most $k-1$ agents. To construct such a configuration, let $j$ be maximal s.t.\ $2\sum_{i=1}^{j-1}N_i\le m$. (We remark that $j\in\{1,...,n\}$, due to $m<k$.) We now start with a $(j-1)$-proper and $j$-empty configuration $C$, and distribute the remaining $m-|C|\le 2N_j$ units evenly across $\overline{x}_j$ and $\overline{y}_j$. The resulting configuration is $j$-low and $(j+1)$-empty.

Regarding the size bound, note that we have $4n+1$ registers. We also have $\O(n)$ instructions: \ref{alg:main} has $\O(n)$ instructions and exists only once, while every other procedure has constant length and is instantiated $\O(n)$ times. The swap-size is $\O(n)$ as well, as only registers $x$ and $\overline{x}$ are swapped, for $x\in\bigcup_iQ_i$.
\end{proof}

\section{Detailed Conversion of Population Programs}\label{app:detailedconversion}
Our goal is to prove the following theorem:

\thmconversion*
\begin{proof}
This will follow from propositions~\ref{pro:high-to-low} and~\ref{pro:low-to-pp}, which are proved in the following two sections.
\end{proof}

\subsection{Semantics of Population Machines}\label{app:machinesemantics}
We start by giving a precise definition of how population machines operate. (An intuitive description can be found in Section~\ref{ssec:formal-model}.)

\begin{definition}\label{def:pm-semantics}
A \emph{configuration} is a map $C$ with $C(x)\in\N$ for $x\in Q$ and $C(X)\in\Flagdom_X$ for $X\in\Flags$. The \emph{output} of $C$ is $C(\OF)$. A configuration $C$ is \emph{initial} if $C(\IP)=1$ and $C(\Virt{x})=x$ for $x\in Q$. For two configurations $C,C'$ we write $C\rightarrow C'$ if
\begin{itemize}
\item $\Inst_{C(\IP)}=(x\mapsto y)$, $C'(\IP)=C(\IP)+1$, $C'(C(\Virt{x}))=C(C(\Virt{x}))-1$, $C'(C(\Virt{y}))=C(C(\Virt{y}))+1$ and $C'(z)=C(z)$ for $z\notin\{\IP, C(\Virt{x}), C(\Virt{y})\}$,
\item $\Inst_{C(\IP)}=(\Maybe\ x>0)$, $C'(\IP)=C(\IP)+1$, $C'(\CF)\in\{\False,C(C(\Virt{x}))>0\}$ and $C'(z)=C(z)$ for $z\notin\{\IP,\CF\}$,
\item $\Inst_{C(\IP)}=(X:=f(Y))$, $X\ne\IP$, $C'(\IP)=C(\IP)+1$, $C'(X)=f(C(Y))$ and $C'(z)=C(z)$ for $z\notin\{\IP,X\}$, or
\item $\Inst_{C(\IP)}=(\IP:=f(Y))$, $C'(\IP)=f(C(Y))$ and $C'(z)=C(z)$ for $z\notin\{\IP\}$.
\end{itemize}
To make the $\rightarrow$ relation left-total, we also define $C\rightarrow C$ if there is no $C'\ne C$ with $C\rightarrow C'$.
\end{definition}

The above definition allows for the computation to “hang” in certain situations, e.g.\ when executing $x\mapsto y$ while $x$ is $0$. If this happens, the computation enters an infinite loop and makes no progress.

We use the general definitions of stable computation from Section~\ref{sec:preliminaries}. We say that $\Machine$ \emph{decides} a predicate $\varphi(x)$ if every fair run starting at an initial configuration $C$ stabilises to $\varphi(\sum_{q\in Q}C(q))$.

\subsection{From Population Programs to Machines}\label{app:detailedmachines}
Let $\Prog=(Q,\Procs)$ denote a population program, we convert it to a population machine $\Machine=(Q,\Flags,\Flagdom,\Inst)$.

\parag{If and while} Our model allows for direct manipulation of the instruction pointer. We use this to implement both conditional and unconditional jumps. To evaluate branches, we use the $\CF$ pointer to store the intermediate boolean results. An example is given in Figure~\ref{fig:convert-while}. For more complicated boolean formulae one needs multiple jumps.

Recall also that for-loops are only a macro in population programs, so we do not have to implement them here.

\begin{figure}[H]
\begin{center}
\begin{minipage}{0.4\textwidth}
\begin{algorithmic}
\While{$\neg(\Maybe\ x>0)$}
    \State $x\mapsto y$
\EndWhile
\State ...
\end{algorithmic}
\end{minipage}\scalebox{2}{$\rightsquigarrow$}\hspace{6mm}
\begin{minipage}{0.4\textwidth}
\begin{algorithmic}[1]
\State \Maybe\ $x>0$
\State $\IP:=\Big\{\begin{array}{ll}
5&\text{ if $\CF$}\\
3&\text{ else}
\end{array}$
\State $x\mapsto y$
\State $\IP:=1$
\State ...
\end{algorithmic}
\end{minipage}
\end{center}
\caption{Implementation of a while-loop.}\label{fig:convert-while}%
\end{figure}

\parag{Procedure calls} In a population program, procedures cannot be recursive. More precisely, the directed graph of calls is acyclic. Recall also that procedures do not take arguments, instead the parameters specify a family of procedures. To take an example from Section~\ref{sec:details}, \ref{alg:checkproper} is not a procedure, but $\ref{alg:checkproper}(1),...,\ref{alg:checkproper}(n)$ are. Hence our implementation only needs to deal with returning from a procedure, which involves jumping to the correct instruction and propagating the return value.

For the former, we use a pointer $P$ for each procedure $P\in\Procs$. This pointer has domain $\Flagdom_P\subseteq\{1,...,L\}$. Calling a procedure involves setting this pointer to the address the procedure should return to, before jumping to the first instruction of the procedure. To propagate return values, we store them in $\CF$. A simple example is shown in Figure~\ref{fig:convert-proc}. While $\Flagdom_P:=\{1,...,L\}$ would work, we limit $\Flagdom_P$ to contain only the necessary elements (i.e.\ one per call of $P$) to reduce the size of the resulting machine.

The population program is specified to start by executing $\Call{Main}{}$, so we insert a call to it as the first instruction followed by an infinite loop in case $\Call{Main}{}$ returns.

\begin{figure}[H]
\begin{center}
\begin{minipage}{0.35\textwidth}
\begin{algorithmic}
\State\Call{AddTwo}{}
\State ...\smallskip
\Procedure{AddTwo}{}
    \State $x\mapsto y$
    \State $x\mapsto y$
    \State\Return\True
\EndProcedure
\end{algorithmic}
\end{minipage}\scalebox{2}{$\rightsquigarrow$}\hspace{6mm}
\begin{minipage}{0.3\textwidth}
\begin{algorithmic}[1]
\State $\Call{AddTwo}{}:=3$
\State $\IP:=4$
\State ...\smallskip
\State $x\mapsto y$
\State $x\mapsto y$
\State $\CF:=\True$
\State $\IP:=\Call{AddTwo}{}$
\end{algorithmic}
\end{minipage}
\end{center}
\caption{Implementation of a procedure.}\label{fig:convert-proc}%
\end{figure}

\parag{Swaps} Most of the heavy lifting is in the definition of the machine model (and the later conversion to population protocols). To implement $(\Swap\ x,y)$ we replace it by the instructions $(\Virt{\square}:=\Virt{x}; \Virt{x}:=\Virt{y}; \Virt{y}:=\Virt{\square})$, which adjust the register map. Similar to procedure calls, we prune $\Flagdom_{\Virt{x}}$ to contain only necessary elements to reduce size; the sum $\sum_{x\in Q}\Abs{\Flagdom_{\Virt{x}}}$ then matches the swap-size introduced in Section~\ref{sec:population-programs}.

\parag{Restarts} A restart changes registers arbitrarily and then continues execution at the beginning. We first transform the population program so that it does the first part by itself, as sketched in Figure~\ref{fig:convert-restart}. Afterwards, the remaining restart instruction (e.g.\ Line 7 in Figure~\ref{fig:convert-restart}) is converted to $\IP:=1$. (One could reset the register map by executing $\Virt{x}:=x$ for $x\in Q$, but this is not necessary as it is always a permutation.)

\begin{figure}[H]
\begin{center}
\hfill\begin{minipage}{0.2\textwidth}
\begin{algorithmic}
\State\Restart
\State ...
\end{algorithmic}
\end{minipage}\scalebox{2}{$\rightsquigarrow$}\hspace{10mm}
\begin{minipage}{0.5\textwidth}
\begin{algorithmic}[1]
\State\Call{Restart}{}
\State ...\smallskip
\Procedure{Restart}{}
    \For{$(y,z)\in Q\times \{x\}\cup \{x\}\times Q$}
        \While{\Maybe\ $y>0$}
            \State $y\mapsto z$
        \EndWhile
    \EndFor
    \State\Restart
\EndProcedure
\end{algorithmic}
\end{minipage}
\end{center}
\caption{Implementing restarts. As an intermediate step, restarts are replaced by a helper procedure that moves to a new configuration before restarting. Here $x\in Q$ is arbitrary.}\label{fig:convert-restart}%
\end{figure}

To summarise, we end up with the following statement.

\begin{restatable}{proposition}{prohightolow}\label{pro:high-to-low}
Let $k\in\N$. If a population program deciding $\varphi$ with size $\lambda$ exists, then there is a population machine deciding $\varphi$ with size $\O(\lambda)$.
\end{restatable}
\begin{proof}
Recall that the size of a population program is $\lambda=n+L+S$, where $n$ is the number of registers, $L$ the number of instructions, and $S$ the swap-size.

Our conversion has exactly $n$ registers. We create a pointer for each register and each procedure, so the number of pointers is $\O(n+L)$. As the pointer domains of the procedure pointers correspond to the call-sites of the respective procedures, the total size of these domains is $\O(L)$. The total size of the domains of the register pointers corresponds to the swap-size, so it is $\O(S)$. The domains of the three special pointers $\OF,\CF$ and $\IP$ have size $\O(L)$.

To estimate the number of instructions note that all instructions, except for \Restart, expand to a constant number if instructions. (Conditionals of while and if statements might be arbitrarily long, but they evaluate a corresponding number of instructions.) For \Restart\ we need to introduce the helper procedure of length $\Theta(n)$, but this overhead is only incurred once. So in total we end up with $\O(n+L)$ instructions.
\end{proof}

\subsection{Conversion to Population Protocols}\label{app:detailedprotocols}
Let $\Machine=(Q,\Flags,\Flagdom,\Inst)$ denote a population machine. Our goal is to convert $\Machine$ to a population protocol $\Prot=(Q^*,\delta,\PPInput,O)$.

\parag{States}
The register agents use states $Q$, while the pointer agent for pointer $X\in\Flags$ uses states of the form $Q_X:=\{X^v_s:v\in\Flagdom_X,s\in \Fstates_X\}$. Here, $v\in\Flagdom_X$ stores the current value of the pointer, while $s\in\Fstates_X$ indicates intermediate stages during the execution of an instruction. The possible values of $s$ depend on the type of pointer:
\[\begin{array}{rll}
\Fstates_\IP&:=\{\Fstate{none},\Fstate{wait},\Fstate{half}\}&\\
\Fstates_X&:=\{\Fstate{none},\Fstate{done},\Fstate{emit},\Fstate{take},\Fstate{test},\True,\False\}&\qquad\text{if $X=\Virt{x}$}\\
\Fstates_X&:=\{\Fstate{none},\Fstate{done}\}&\qquad\text{if $X\ne\Virt{x}$, $X\ne\IP$}
\end{array}\]

Finally, to perform the mappings necessary for instruction of the form $(X:=f(Y))$, we add states $Q_\mathrm{map}:=\{X^i_{\Fstate{map}}:\Inst_i=(X:=f(Y))\}$.

In total, we have states $Q^*:=Q\cup\bigcup_{X\in\Flags}Q_X\cup Q_\mathrm{map}$.

\parag{Initial states and leader election}
Let $X_1,...,X_{\Abs{\Flags}}$ denote some enumeration of $\Flags$ with $X_{\Abs{\Flags}}=\IP$. We set $\PPInput:=\{X_1\}$, i.e.\ we use $X_1$ as unique initial state.

For each pointer $X_i$, fix an initial value $v_i\in\Flagdom_{X_i}$. These initial values must fulfil the requirements of initial configurations set forth in Definition~\ref{def:pm-semantics}, i.e.\ $v_{\Abs{\Flags}}:=1$ (recall $X_{\Abs{\Flags}}=\IP$), and $v_i:=x$ if $X_i=\Virt{x}$ for $x\in Q$. To define the transitions, we also fix some arbitrary register $x\in Q$. For convenience, we use $*$ as a wildcard.
\[\renewcommand{\arraystretch}{1.1}
\begin{array}{rll}
(X_i)^*_*, (X_i)^*_*&\mapsto (X_i)^{v_i}_{\Fstate{none}}, (X_{i+1})^{v_{i+1}}_{\Fstate{none}}&\qquad\text{for $i=1,...,\Abs{\Flags}-1$}\\
\IP^*_*, \IP^*_*&\mapsto (X_1)^{v_1}_{\Fstate{none}}, x
\end{array}\TraName{elect}\]
Intuitively, whenever two agents in $X_i$ meet, one of them moves to $X_{i+1}$, initialising it in the process. The pointer $\IP$ is handled slightly differently: here one of the agents moves to $x$ and thus becomes a register agent, while the other moves to $X_1$. This will then re-initialise $X_1,...,X_{\Abs{\Flags}}$.

\parag{Instructions}
The transitions for executing an instruction $I_i$, $i\in\{1,...,L\}$, depend on the type of instruction. The first case is $I_i=(x\mapsto y)$. This is somewhat involved as we need to first translate $x$ and $y$ using the register map. First (the agent responsible for) $\IP$ instructs $\Virt{x}$ to move one agent from the register currently assigned to $x$ to some fixed register $z$. (Note that $z$ is independent of the instruction.) After that is completed, $\Virt{y}$ moves the agent from $z$ to its target. Note that $i=L$ means that the machine hangs.
\[\renewcommand{\arraystretch}{1.2}
\begin{array}{llllll}
\IP^{\,i}_{\Fstate{none}},&(\Virt{x})^v_*&\mapsto \IP^{\,i}_{\Fstate{wait}},&(\Virt{x})^v_{\Fstate{emit}}&\qquad\text{for $v\in\Flagdom_{\Virt{x}}$}\\
(\Virt{x})^v_{\Fstate{emit}},&v&\mapsto (\Virt{x})^v_{\Fstate{done}},&z&\qquad\text{for $v\in\Flagdom_{\Virt{x}}$}\\
\IP^{\,i}_{\Fstate{wait}},&(\Virt{x})^v_\Fstate{done}&\mapsto \IP^{\,i}_{\Fstate{half}},&(\Virt{x})^v_{\Fstate{none}}&\qquad\text{for $v\in\Flagdom_{\Virt{x}}$}\\
\IP^{\,i}_{\Fstate{half}},&(\Virt{y})^v_*&\mapsto \IP^{\,i}_{\Fstate{wait}},&(\Virt{y})^v_{\Fstate{take}}&\qquad\text{for $v\in\Flagdom_{\Virt{y}}$}\\
(\Virt{y})^v_{\Fstate{take}},&z&\mapsto (\Virt{y})^v_{\Fstate{done}},&v&\qquad\text{for $v\in\Flagdom_{\Virt{y}}$}\\
\IP^{\,i}_{\Fstate{wait}},&(\Virt{y})^v_\Fstate{done}&\mapsto \IP^{\,i+1}_{\Fstate{none}},&(\Virt{y})^v_{\Fstate{none}}&\qquad\text{if $i<L$, for $v\in\Flagdom_{\Virt{x}}$}
\end{array}\TraName{move}\]

For $I_i=(\Maybe\ x>0)$ the $\IP$ agent again recruits the $\Virt{x}$ agent to do the actual operation. The latter either detects $x$ or it does not, and then stores the result in $\CF$.
\[\renewcommand{\arraystretch}{1.2}
\begin{array}{llllll}
\IP^{\,i}_{\Fstate{none}},&(\Virt{x})^v_*&\mapsto \IP^{\,i}_{\Fstate{wait}},&(\Virt{x})^v_{\Fstate{test}}&\qquad\text{for $v\in\Flagdom_{\Virt{x}}$}\\
(\Virt{x})^v_{\Fstate{test}},&v&\mapsto (\Virt{x})^v_{\True},&v&\qquad\text{for $v\in\Flagdom_{\Virt{x}}$}\\
(\Virt{x})^v_{\Fstate{test}},&q&\mapsto (\Virt{x})^v_{\False},&q&\qquad\text{for $v\in\Flagdom_{\Virt{x}}$, $q\in Q^*\setminus\{v\}$}\\
(\Virt{x})^v_b,&\CF^*_*&\mapsto (\Virt{x})^v_{\Fstate{done}},&\CF^b_\Fstate{none}&\qquad\text{for $v\in\Flagdom_{\Virt{x}}$, $b\in\{\True,\False\}$}\\
\IP^{\,i}_{\Fstate{wait}},&(\Virt{x})^v_\Fstate{done}&\mapsto \IP^{\,i+1}_{\Fstate{none}},&(\Virt{x})^v_{\Fstate{none}}&\qquad\text{if $i<L$, for $v\in\Flagdom_{\Virt{x}}$}
\end{array}\TraName{test}\]
The third type, $I_i=(X:=f(Y))$, has some special cases. We first assume $Y\ne\IP$ wlog, as the value of $\IP$ is simply $i$ and $f(Y)$ could be replaced by a constant expression. Both $X=Y$ and $X=\IP$ have to be handled separately. The general procedure then is that (the agent responsible for) $\IP$ moves $X$ into an intermediate state in $Q_\mathrm{map}$ and waits. Then, $X$ meets $Y$, updates its value, and finally signals $\IP$ to continue to computation.

We start with the ordinary case $X\notin\{Y,\IP\}$.
\[\renewcommand{\arraystretch}{1.2}
\begin{array}{rll}
\IP^{\,i}_{\Fstate{none}},X^*_*&\mapsto\IP^{\,i}_{\Fstate{wait}},X^i_{\Fstate{map}}&\qquad\text{if $i<L$}\\
X^i_{\Fstate{map}},Y^v_*&\mapsto X^{f(v)}_{\Fstate{done}}, Y^v_{\Fstate{none}}&\qquad\text{for $v\in\Flagdom_Y$}\\
\IP^{\,i}_{\Fstate{wait}},X^{v}_{\Fstate{done}}&\mapsto\IP^{\,i+1}_{\Fstate{none}},X^{v}_{\Fstate{none}}&\qquad\text{for $v\in\Flagdom_X$}
\end{array}\TraName{pointer}\]
Now we handle the special cases. These are easier, as only two agents are involved.
{\renewcommand{\TraNs}{temp}
\[\renewcommand{\arraystretch}{1.2}
\begin{array}{rll}
\IP^{\,i}_{\Fstate{none}},Y^v_*&\mapsto\IP^{f(i)}_{\Fstate{none}},Y^v_{\Fstate{none}}&\qquad\text{if $X=\IP$, for $v\in\Flagdom_Y$}\\
\IP^{\,i}_{\Fstate{none}},Y^v_*&\mapsto\IP^{\,i+1}_{\Fstate{none}},Y^{f(v)}_{\Fstate{none}}&\qquad\text{if $X=Y$, $i<L$, for $v\in\Flagdom_Y$}
\end{array}\TraName{pointer}\]}

\parag{Output broadcast} As mentioned above, we need to ensure that the agents come to a consensus. So we convert $\Prot$ again, to the final population protocol $\Prot'=(Q',\delta',\PPInput',O')$. This uses the standard broadcast construction, so $Q':=Q^*\times\{\True,\False\}$, $\PPInput':=\PPInput\times\{\False\}$, $O':=Q'\times\{\True\}$ and for all $q_1,q_2,q_1',q_2'\in Q^*$ with $(q_1,q_2\mapsto q_1',q_2')\in\delta$ or $(q_1,q_2)=(q_1',q_2')$ we have transitions
\[\begin{array}{rll}
(q_1,*),(q_2,*)&\mapsto (q_1',b),(q_2',b)&\qquad\text{if $\OF^b_*\in\{q_1',q_2'\}$, for a $b\in\{\True,\False\}$}\\
(q_1,b_1),(q_2,b_2)&\mapsto (q_1',b_1),(q_2',b_2)&\qquad\text{otherwise}
\end{array}\]

\parag{Correctness}
We now show that the above conversion is correct. We first define a mapping $\pi$ between configurations of the population machine $\Machine$ and the population protocol $\Prot$ resulting from our conversion. A configuration $C$ of $\Machine$ is mapped to a configuration $\pi(C)$ of $\Prot$ as follows.
\[\begin{array}{lll}
\pi(C)(x)&:=C(x)&\qquad\text{for $x\in Q$}\\
\pi(C)(X^v_{\Fstate{none}})&:=1&\qquad\text{if $C(X)=v$, for $X\in\Flags,v\in\Flagdom_X$}\\
\pi(C)(X^*_*)&:=0&\qquad\text{otherwise}\\
\end{array}\]

First, we prove that any configuration with sufficiently many agents in the initial state reaches a configuration $\pi(C)$, for some $C$. \footnote{To show correctness, we need only the case $c\in\N^{\PPInput}$, but we use it also to show almost self-stabilisation.}

\begin{lemma}\label{lem:electworks}
Every configuration $c\in\N^{Q^*}$ with $c(\PPInput)\ge\Abs{F}$ reaches $\pi(C)\in\N^{Q^*}$ for some initial configuration $C$ of $\Machine$ with $|C|=|\pi(C)|-\Abs{\Flags}$.
\end{lemma}
\begin{proof}
Let $X_1,...,X_{\Abs{\Flags}}$ denote the enumeration used for \TraRef{elect}, and let $d\in\N^{Q^*}$ denote a configuration. If we consider the tuple $\big(d(Q), d((X_{\Abs{\Flags}})^*_*),...,d((X_1)^*_*),\big)$, we see that executing \TraRef{elect} increases its value lexicographically. Hence \TraRef{elect} can only be executed finitely often.

Let $c'$ denote any configuration reachable by $c$. If $c'(X^*_*)\ge2$ for some $X\in\Flags$, then \TraRef{elect} can be executed, so eventually we reach a configuration $c'$ with $c'(X^*_*)\le1$ for all $X$.

Now we use $c(\PPInput)\ge\Abs{F}$. By a simple induction we observe that for every $i\le\Abs{\Flags}$ we have $c'((X_1)^*_*)+...+c'((X_i)^*_*)\ge i$ for every configuration $c'$ reachable from $c$. So eventually, there is exactly one agent in $X^*_*$ for all $X\in\Flags$. At the moment this happens, these agents are in $X^v_{\Fstate{none}}$, where $v$ is the initial state of the pointer. Therefore we have reached a configuration $\pi(C)$; moreover, $C$ must be an initial configuration of $\Machine$ with $|C|=|\pi(C)|-\Abs{\Flags}$ agents.
\end{proof}

\begin{restatable}{proposition}{prolowtopp}\label{pro:low-to-pp}
If a population machine deciding $\varphi$ with size $n$ exists, then there is a population protocol deciding $\varphi'(x)\Leftrightarrow \varphi(x-i)\wedge x\ge i$ with $\O(n)$ states, for some $i\le n$.
\end{restatable}
\begin{proof}
If $\Prot$ is run on a configuration with fewer than $\Abs{\Flags}$ agents, no agent can reach a state $\IP^*_*$ via \TraRef{elect}, and no other transition is enabled. In particular, it is not possible for any agent to enter $\OF^{\True}_*$.

If at least $\Abs{\Flags}$ agents are present, then we use Lemma~\ref{lem:electworks} to show that we eventually reach a configuration $\pi(C)$, where $C$ is initial and $|C|=|\pi(C)|-\Abs{\Flags}$.

To see that a run of $\Prot$ corresponds to one of $\Machine$, we need only convince ourselves that \TraRef{move}, \TraRef{test} and \TraRef{pointer} correctly implement the semantics of Definition~\ref{def:pm-semantics} and move to a configuration $\pi(C')$, where $C\rightarrow C'$.

Every fair run of $\Machine$ stabilises a $b\in\{\True,\False\}$, according to $\varphi$. So eventually there will be a unique agent in $\OF^b_*$, and it will remain in one of these states.

It remains to argue that runs of $\Prot'$ correspond to runs of $\Prot$ (and thus to runs of $\Machine$), and that they stabilise to the correct output. The former is easy to see, as the output broadcast construction simply uses the first component to execute $\Prot$ (and this is not affected by the second). Once a unique agent remains in $\OF^b_*$ in $\Prot$, the corresponding run in $\Prot'$ will have an agent in $(\OF^b_*,b)$. Eventually, this agent will convince all other agents that the output is $b$, and the computation stabilises to $b$. 

As $\Prot$ (and $\Prot'$) use $\Abs{\Flags}$ agents to store the value of each pointer, the corresponding configurations of $\Machine$ are smaller, and $\Prot'$ decides $\varphi'(x)\Leftrightarrow x\ge\Abs{\Flags}\wedge\varphi(x-\Abs{\Flags})$.

Finally, we need to count the states of $\Prot'$. We have $\Abs{Q'}=2\cdot\Abs{Q^*}$ and
\[\Abs{Q^*}=\Abs{Q}+\sum_{X\in\Flags}\Abs{Q_X}+\Abs{Q_\mathrm{map}}\le\Abs{Q}+7\sum_{X\in\Flags}\Abs{\Flagdom_X}+L\in\O(n)\]
\end{proof}

\end{document}